\documentclass[11pt]{article}

\usepackage{a4wide}
\usepackage{color}
\usepackage{amsfonts,amsmath,amsthm}
\usepackage{authblk}

\newcommand{\eps}{\varepsilon}

\newcommand{\ex}{\mathbb{E}}

\newcommand{\B}{{\mathcal B}}

\newcommand{\E}{{\mathcal E}}

\newcommand{\C}{{\textsc C}}

\newcommand{\la}{\ell}

\newcommand{\Bin}{\mathrm{Bin}}

\newcommand{\CE}{{\mathcal{E}}}

\renewcommand{\Pr}[1]{\mathbb{P}\left(#1\right)}
\renewcommand{\E}[1]{\mathbb{E}\left(#1\right)}

\newcommand{\Pcl}{\widetilde{H}_{n,p,k}}

\newtheorem{theorem}{Theorem}[section]

\newtheorem{corollary}[theorem]{Corollary}
\newtheorem{proposition}[theorem]{Proposition}
\newtheorem{lemma}[theorem]{Lemma}
\newtheorem{claim}[theorem]{Claim}

\numberwithin{equation}{section}
\numberwithin{firsttheorem}{section}

\title{The Multiple-orientability Thresholds for Random Hypergraphs\footnote{An extended abstract of this work appeared in the \emph{Proceedings of the 22nd ACM-SIAM Symposium on Discrete Algorithms (SODA '11).}}}
\author[1]{Nikolaos Fountoulakis}
\author[2]{Megha Khosla}
\author[3]{Konstantinos Panagiotou}
\affil[1]{School of Mathematics, University of Birmingham, United Kingdom}
\affil[2]{Max Planck Institute for Informatics Saarbr\"ucken, Germany}
\affil[3]{Mathematisches Institut Ludwig-Maximilians-Universit\"at M\"unchen, Germany}

\date{\vspace{-5ex}}
\begin{document}
\maketitle

\begin{abstract} \small
A $k$-uniform hypergraph~$H = (V, E)$ is called~$\ell$-orientable, if there is an assignment of each edge $e\in E$ to one of its vertices $v\in e$ such that no vertex is assigned more than~$\ell$ edges. Let~$H_{n,m,k}$ be a hypergraph, drawn uniformly at random from the set of all~$k$-uniform hypergraphs with~$n$ vertices and~$m$ edges. In this paper we establish the threshold for the $\ell$-orientability of~$H_{n,m,k}$ for all~$k\ge 3$ and $\ell \ge 2$, i.e., we determine a critical quantity~$c_{k, \ell}^*$ such that with probability~$1-o(1)$ the graph~$H_{n,cn,k}$ has an $\ell$-orientation if~$c < c_{k, \ell}^*$, but fails doing so if~$c > c_{k, \ell}^*$.

Our result has various applications including sharp load thresholds for cuckoo hashing, load balancing with guaranteed maximum load, and massive parallel access to hard disk arrays.
\end{abstract}

\section{Introduction}

This paper studies the property of multiple orientability of random hypergraphs. For any integers $k \ge 2$ and $\ell \ge 1$, a $k$-uniform hypergraph is called $\ell$\emph{-orientable}, if for each edge we can select one of its vertices, so that all vertices are selected at most $\ell$ times. This definition generalizes the classical notion of orientability of graphs, where we want to orient the edges under the condition that no vertex has in-degree larger than $\ell$. In this paper, we consider random $k$-uniform hypergraphs $H_{n,m,k}$, for $k\geq 3$, with~$n$ vertices and~$m = \lfloor cn \rfloor$ edges. Our main result establishes the existence of a critical density $c_{k, \ell}^*$ (determined explicitly in Thorem~\ref{thm:main}), such that when $c$ crosses this value the probability that the random hypergraph is $\ell$-orientable drops abruptly from $1-o(1)$ to $o(1)$, as the number of vertices $n$ grows.  

The case $k=2$ and $\ell \ge 1$ is well-understood. In fact, this case corresponds to the classical random graph $G_{n,m}$ drawn uniformly from the set of all graphs with $n$ vertices and $m$ edges. A result of Fernholz	and Ramachandran \cite{1283432} and Cain, Sanders and Wormald~\cite{1283433} implies that there is a constant 
$c_{2, \ell}^\ast$ such that as $n \to \infty$
\[
	\Pr{G_{n,\lfloor cn\rfloor}\text{ is $\ell$-orientable}}
	\to
	\begin{cases}
		0, & \text{ if } c > c_{2, \ell}^\ast \\
		1, & \text{ if } c < c_{2, \ell}^\ast
	\end{cases}.
\]
In other words, there is a critical value such that when the average degree is below this, then with high probability an $\ell$-orientation exists, and otherwise not. We want to remark at this point that the orientation can be found efficiently by solving a matching problem on a suitably defined bipartite graph, but we will not consider computational issues any further in this paper.

Similarly, the case~$\ell = 1$ and $k \ge 3$ arbitrary is also well-understood. The threshold for the 1-orientabilty is known from the work of the first and the third author~\cite{fp10,fp12}, and Frieze and Melsted~\cite{fm12}. In particular, there is a constant~$c_{k, 1}^\ast$ such that as $n \to \infty$
\[
	\Pr{H_{n,\lfloor cn\rfloor, k}\text{ is~$1$-orientable}}
	\to
	\begin{cases}
		0, & \text{ if } c > c_{k, 1}^\ast \\
		1, & \text{ if } c < c_{k, 1}^\ast
	\end{cases}.
\]
In this paper we consider the general case, i.e.,~$k$ and~$\ell$ arbitrary. Our main result is summarized in the following theorem, and settles the threshold for the $\ell$-orientability property of random hypergraphs for all $k$ and $\ell$.
\begin{theorem}
\label{thm:main}
For integers~$k\geq 3$ and~$\ell \geq 2$ let~$\xi^\ast$ be the unique solution of the equation
\begin{equation}\label{eq:kxi}
 k \ell = \frac{\xi^\ast Q(\xi^\ast,\ell)}{Q(\xi^\ast,\ell+1)}, \text{ where } Q(x,y) = 1 - e^{-x}\sum_{j < y}\frac{x^j}{j!}.
\end{equation}
Let~$c_{k, \ell}^\ast= \frac{\xi^\ast}{kQ(\xi^\ast,\ell)^{k-1}}$. Then
\begin{equation}
\label{eq:phnmkl}
	\Pr{H_{n,\lfloor cn\rfloor, k}\text{ is~$\ell$-orientable}}
	\stackrel{(n \to \infty)}{=} 
	\begin{cases}
		0, & \text{ if } c > c_{k, \ell}^\ast \\
		1, & \text{ if } c < c_{k, \ell}^\ast
	\end{cases}.
\end{equation}
\end{theorem}
A similar result by using completely different techniques was also shown recently in a slightly different context by Gao and 
Wormald~\cite{gw10}, with the restriction that the product~$k \ell$ is large. So, our result fills the remaining gap, and treats especially the
cases of small~$k$ and arbitrary~$\ell$, which are most interesting in practical applications. Further generalizations of the concept of
orientability of hypergraphs have been considered after our work in~\cite{inp:Lelarge} and \cite{inp:LecLelargeMas}, where tight results are also obtained.

\subsection{Applications}

\paragraph{Cuckoo Hashing} The paradigm of many choices has influenced significantly the design of efficient data structures and, most notably, hash tables. \emph{Cuckoo hashing}, introduced by Pagh and Rodler~\cite{inp:pr01}, is a technique that extends this concept.
We consider here a slight variation of the original idea, see also the paper~\cite{inc:fpss03} by Fotakis, Pagh,  Sanders and Spirakis, where we are given a table with~$n$ locations, and we assume that each location can hold $\ell$ items. 
Each item to be inserted chooses randomly $k\ge 2$ locations and has to be placed in any one of them. 
How much load can cuckoo hashing handle before collisions make the successful assignment of the available items 
to the chosen locations impossible? 
Practical evaluations of this method have shown that one can allocate a number of 
elements that is a large proportion of the size of the table, being very close to~1 even for small values of~$k\ell$ such as 4 or 6. Our main theorem provides the theoretical foundation for this empirical observation: if the number of items is less than $c_{k, \ell}^* n$, then it is highly likely that they can be allocated. However, if their number is larger, then most likely every allocation will have an overfull bin. Our result thus proves a conjecture about the threshold loads of cuckoo hashing made in~\cite{Dietz09}.
\paragraph{Load Balancing}

In a typical load balancing problem we are given a set of $m = \lfloor cn \rfloor$ identical jobs, and $n$ machines on which they can be executed. Suppose that each job may choose randomly among $k$ different machines. Is there any upper bound for the maximum load that can be guaranteed with high probability? Our main result implies that as long as $c < c^\ast_{k,\ell}$, then there is an assignment of the jobs to their preferred machines such that no machine is assigned more than $\ell$ different tasks.

\paragraph{Parallel Access to Hard Disks}

In our final application we are given $n$ hard disks (or any other means of storing large amounts of information), which can be accessed independently of each other. We want to  store there a big data set redundantly, that gives us some degree of fault tolerance, and at the same time minimize the number of I/O steps needed to retreive the data (see \cite{sek99} for more details). Theorem~\ref{thm:main} implies that if $k$ randomly allocated copies of each block exist on $n$ hard disks then $m=\lfloor cn \rfloor$ different data blocks can be read, with at most $\ell$ parallel queries on each disk, with high probability provided that $c < c^\ast_{k,\ell}$.

\section{Proof Strategy \& The Upper Bound}
\label{sec:prf_str}
Our main result follows immediately from the two theorems below. The first statement says that~$H_{n,m,k}$ has a subgraph of density~$>\ell$ (i.e., the ratio of the number of edges to the number of vertices in this subgraph is greater than~$\ell$) if~$c > c^\ast_{k, \ell}$. 
We denote by the~$(\ell+1)$-core of a hypergraph its maximum subgraph that has minimum degree at least~$\ell+1$.
\begin{theorem}\label{thm:simple}
 Let~$c_{k,\ell}^*$ be defined as in Theorem~\ref{thm:main}. If~$c>c_{k,\ell}^\ast$, then with probability~$1-o(1)$ the~$(\ell+1)$-core of~$H_{n,cn,k}$ has density greater than~$\ell$.
\end{theorem}
Note that this implies the statement in the first line of~\eqref{eq:phnmkl}, as by the pigeonhole principle it is impossible to orient the edges of a hypergraph with density larger than~$\ell$ so that each vertex has indegree at most~$\ell$.

The above theorem is not very difficult to prove, as the core of random hypergraphs and its structural characteristics have been studied quite extensively in recent years, see e.g.\ the results by Cooper~\cite{ar:c04}, Molloy \cite{ar:m05} and Kim~\cite{ar:k06}. However, it requires some technical work, which is accomplished in Section~\ref{ssec:ProofThmSimple}. The heart of this paper is devoted to the ``subcritical" case, where we show that the above result is essentially tight.
\begin{theorem}\label{thm:diff}
 Let~$c_{k,\ell}^*$ be defined as in Theorem~\ref{thm:main}. If~$c<c_{k,\ell}^\ast$, then with probability~$1-o(1)$ all subgraphs of~$H_{n,cn,k}$ have density smaller than~$\ell$.
\end{theorem}
\begin{proof}[Proof of Theorem~\ref{thm:main}.]
Let us construct an auxiliary bipartite graph~$B=(\mathcal{E}, \mathcal{V};\, E)$, where~$\mathcal{E}$ represents the~$m$ edges and~$\mathcal{V} = \{1, \dots, n\} \times \{1, \dots, \ell\}$ represents the~$n$ vertices of~$H_{n,m,k}$. Also,~$\{e,(i, j)\}\in E$ if the~$e$th edge contains vertex~$i$, and $1 \le j \le \ell$. Note that $H_{n,m,k}$ is $\ell$-orientable if and only if $B$ has a left-perfect matching, and by Hall's theorem such a matching exists if and only if for all $\mathcal{I} \subseteq \mathcal{E}$ we have that~$|\mathcal{I}|\leq |\Gamma(\mathcal{I})|$, where~$\Gamma(\mathcal{I})$ denotes the set of neighbors of the vertices in~$\mathcal{I}$ in~$\mathcal{V}$.

Observe that~$\Gamma(\mathcal{I})$ is precisely the set of~$\ell$ copies of the vertices that are contained in the hyperedges corresponding to items in~$\mathcal{I}$. So, if~$c<c_{k,\ell}^\ast$, Theorem~\ref{thm:diff} guarantees that with high probability for all~$\mathcal{I}$ we have~$|\mathcal{I}|\leq |\Gamma(\mathcal{I})|$ and therefore $B$ has a left-perfect matching. On the other hand, if~$c>c^\ast_{k,\ell}$, then with high probability there is a set~$\mathcal{I}$ such that~$|\mathcal{I}| > |\Gamma(\mathcal{I})|$; choose for example~$\mathcal{I}$ to be the set of items that correspond to the edges in the~$(\ell+1)$-core of~$H_{n,m,k}$. Hence a matching does not exist in this case, and the proof is completed. 
\end{proof}

\subsection{Proof of Theorem~\ref{thm:simple} and the Value of $c_{k,\ell}^*$}
\label{ssec:ProofThmSimple}

The aim of this section is to determine the value~$c_{k,\ell}^*$ and prove Theorem~\ref{thm:simple}. Moreover, we will introduce some known facts and tools that will turn out to be very useful in the study of random hypergraphs, and will be used later on in the proof of Theorem~\ref{thm:diff} as well. In what follows we will be referring to a hyperedge of size $k$ as a ($k$-)edge and we will be calling a hypergraph with all its hyperedges of size $k$ a~\emph{$k$-graph}.

\subsubsection*{Models of Random Hypergraphs}

For the sake of convenience we will carry out our calculations in the~$H_{n,p,k}$ model of random~$k$-graphs. 
This is the ``higher-dimensional" analogue of the well-studied~$G_{n,p}$ model, where each possible ($k$-)edge is included independently with probability~$p$. More precisely, given~$n\geq k$ vertices we obtain~$H_{n,p,k}$ by including each~$k$-tuple of vertices with probability~$p$, independently of every other~$k$-tuple. 

Standard arguments show that if we adjust~$p$ suitably, then the~$H_{n,p,k}$ model is essentially equivalent to the~$H_{n,cn,k}$ model. Let us be more precise. Suppose that~${\mathcal P}$ is a {\em convex} hypergraph property, that is, whenever we have three hypergraphs $H_1,H_2,H_3$ such that~$H_1 \subseteq H_2 \subseteq H_3$ and~$H_1, H_3 \in {\mathcal P}$, then also~$H_2 \in {\mathcal P}$. We also assume that~${\mathcal P}$ is closed under automorphisms. Any monotone property is also convex and, therefore, the properties examined in Theorem~\ref{thm:diff}. The following proposition is a generalization of Proposition 1.15 from~\cite[p.16]{b:jlr00} and its proof is very similar to the proof of that -- so we omit it.

\begin{proposition} \label{prop:ModelsEquiv}
Let ${\mathcal P}$ be a convex property of hypergraphs, and let $p = ck/\binom{n-1}{k-1}$, where~$c > 0$. If $\Pr {H_{n,p,k}  \in {\mathcal P}} \rightarrow 1$ as $n \rightarrow \infty$, then $\Pr {H_{n,\lfloor cn\rfloor,k} \in {\mathcal P}} \rightarrow 1$ as well. 
\end{proposition}  

\subsubsection*{Working on the $(\ell +1)$-core of $H_{n,p,k}$ -- the Cloning Model}

Recall that the $(\ell+1)$-core of a hypergraph is its maximum subgraph that has minimum degree (at least)~$\ell+1$. At this point we introduce the main tool for our analysis. The \emph{cloning model} with parameters $(N, D, k)$, where $N$ and $D$ are integer valued random variables, is defined as follows. We generate a graph in three stages.
\begin{enumerate}
\item[1.] We expose the value of $N$;
\item[2.] if $N\ge 1$ we expose the degrees $\mathbf{d} = (d_1, \dots, d_{N})$, where the $d_i$'s are independent samples from the distribution $D$;
\item[3.] for each $1 \le v \le N$ we generate~$d_v$ copies, which we call {\em $v$-clones} or simply {\em clones}. Then we choose uniformly at random a matching from all perfect~$k$-matchings on the set of all clones, i.e., all partitions of the set of clones into sets of size $k$. Note that such a matching may not exist -- in this case we choose a random matching that leaves less than $k$ clones unmatched. Finally, we construct the $k$-graph $H_{\mathbf{d}, k}$ by contracting the clones to vertices, i.e., by projecting the clones of $v$ onto $v$ itself for every $1 \le v \le N$.
\end{enumerate}
Note that the last stage in the above procedure is equivalent to the \emph{configuration model}~\cite{Bol1,BenCan} $H_{\mathbf{d}, k}$ for random hypergraphs with degree sequence~$\mathbf{d}= (d_1, \dots, d_n)$. In other words, $H_{\mathbf{d}, k}$ is a random multigraph where the $i$th vertex has degree $d_i$.
~\\

One particular case of the cloning model is the so-called \emph{Poisson cloning model}~$\widetilde{H}_{n,p,k}$ for~$k$-graphs with~$n$ vertices and parameter~$p \in [0,1]$, which was introduced by Kim~\cite{ar:k06}. There, we choose $N = n$ with probability 1, and the distribution $D$ is the Poisson distribution with parameter~$\lambda := p {n-1 \choose k-1}$. Note that $D$ is essentially the vertex degree distribution in the binomial random graph $H_{n,p,k}$, so we would expect that the two models behave similarly. The following statement confirms this, and is implied by Theorem~1.1 in~\cite{ar:k06}.
\begin{theorem}
\label{cor:Contiguity}
If $\Pr{\widetilde{H}_{n,p,k} \in \mathcal{P}} \rightarrow 0$ as $n \rightarrow \infty$, then
$\Pr{H_{n,p,k} \in \mathcal{P}} \rightarrow 0$ as well.
\end{theorem}
One big advantage of the Poisson cloning model is that it provides a very precise description of the $(\ell+1)$-core of $\widetilde{H}_{n,p,k}$. Particularly, Theorem~6.2 in~\cite{ar:k06} implies the following statement, where we write ``$x \pm y$'' for the interval of numbers $(x - y, x+y)$.
\begin{theorem}\label{thm:CoreClone}
Let $\lambda_{k,\ell+1}:= \min_{x>0} \frac{x}{Q(x,\ell)^{k-1}}$. Assume that $ck=p {n-1\choose k-1}> \lambda_{k,\ell +1}$. Moreover, let~$\bar{x}$ be the largest solution of the equation $x=Q(xck,\ell)^{k-1}$, and set $\xi:= \bar{x}ck$.  Then, for any $0<\delta< 1$ the following is true with probability $1-n^{-\omega(1)}$. If $\tilde{N}_{\ell+1}$ denotes the number of vertices in the $(\ell +1)$-core of $\widetilde{H}_{n,p,k}$, then
 $$ \tilde{N}_{\ell +1}= Q(\xi,\ell+1)n \pm \delta n.$$
Furthermore, the $(\ell +1)$-core itself is distributed like the cloning model with parameters $(\tilde{N}_{\ell +1}, \, \mathrm{Po}_{\ge \ell + 1}(\Lambda_{c,k,\ell}), \, k)$, where $\mathrm{Po}_{\ge \ell + 1}(\Lambda_{c,k,\ell})$ denotes a
Poisson random variable conditioned on being at least $(\ell +1)$ and parameter $\Lambda_{c,k,\ell}$, where $\Lambda_{c,k,\ell}=\xi +\beta$, for some $\beta$ satisfying $|\beta|\leq \delta$. 
\end{theorem}
In what follows, we say that a random variable is an $\ell$-truncated Poisson variable, if it is distributed like a Poisson variable, conditioned on being at least $\ell$. The following theorem, which is a special case of Theorem II.4.I in~\cite{b:e06} from large deviation theory, bounds the sum of i.i.d. random variables. We apply the result to the case of i.i.d.~$(\ell+1)$-truncated Poisson random variables, which are nothing but the degrees of the vertices of the $(\ell +1)$-core. As an immediate corollary we obtain tight bounds on the number of edges in the $(\ell +1)$-core of $\widetilde{H}_{n,p,k}$. Moreover, it also serves as our main tool in counting the expected number of subsets (with some density constraints) of the $(\ell +1)$-core, assuming that the degree sequence has been exposed. Such estimates are required for the proof of Theorem~\ref{thm:diff} and will be presented in the next section.

\begin{theorem} \label{thm:largedeviation} 
Let $X$ be a random variable taking real values and set 
$c(t) = \ln \ex (e^{tX})$, for any $t \in {\mathbb R}$. For any $z>0$ we define 
$I(z) = \sup_{t\in {\mathbb R}} \{ zt - c(t)\}$. If $X_1,\ldots, X_s$ are i.i.d. random variables 
distributed as $X$, then for $s \to \infty$
$$\Pr { \sum_{i=1}^s X_i  \leq sz} = \exp \left( -s \inf \{I(x) : x\leq z \} (1+o(1)) \right). $$ 
The function $I(z)$ is non-negative and convex.
\end{theorem}

The function $I(z)$ (also known as the \emph{rate function} of the random variable $X$) in the above theorem measures the discrepancy between $z$ and the expected value of the sum of the i.i.d.~random variables in the sense that $I(z)\ge 0$ with equality if and only if $z$ equals the expected value of $X$. The following lemma applies Theorem~\ref{thm:largedeviation} to $(\ell+1)$-truncated Poisson random variables.
\begin{lemma}\label{lem:I}
Let $X_1,\ldots,X_s$ be i.i.d. $(\ell+1)$-truncated Poisson random variables with parameter $\Lambda$. For any $z>\ell+1$, let $T_z$ be the unique solution of $z= T_z\cdot \frac{Q(T_z,\ell)}{Q(T_z,\ell+1)}$ and
\begin{align} \label{Iz}
 I_{\Lambda}(z)&=z(\ln T_z-\ln \Lambda)- T_z + \Lambda- \ln Q(T_z,\ell+1)+\ln Q(\Lambda,\ell+1). 
\end{align}
Then $I_{\Lambda}(z)$ is continuous for all $z> \ell+1$ and convex. It has a unique minimum at $z=\mu=\Lambda \cdot \frac{Q(\Lambda,\ell)}{Q(\Lambda,\ell+1)}$, where $I_\Lambda(\mu)=0.$ Moreover uniformly for any $z$ such that $\ell+1\leq z\leq \mu$, we have as $s\rightarrow \infty$
$$ \Pr{ \sum_{i=1}^s X_i\leq sz}\leq \exp(-s I_{\Lambda}(z)(1+o(1))).$$
\end{lemma}
\begin{proof}
We shall first calculate $c(t) = \ln \ex ( e^{tX} )$, where $X$ is an $(\ell+1)$-truncated Poisson random variable with parameter $\Lambda$. We note that
\begin{align*}
 \exp(c(t))&=\frac{{\sum_{j\geq \ell+1}}e^{tj}\cdot \frac{e^{-\Lambda}\Lambda^{j}}{j!}}{Q(\Lambda,\ell+1)}
=~e^{-\Lambda}\cdot e^{\Lambda e^t}\cdot \frac{\sum_{j\geq \ell+1}\frac{e^{-\Lambda e^t}(e^t\Lambda)^j}{j!}}{Q(\Lambda,\ell+1)}= ~e^{\Lambda e^t-\Lambda}\cdot \frac{Q(\Lambda e^t,\ell+1)}{Q(\Lambda,\ell+1)}.
\end{align*}
Differentiating $zt-c(t)$ with respect to $t$ we obtain
\begin{align*}
(zt-c(t))^\prime= ~&z- \ln \left( e^{\Lambda e^t-\Lambda}\cdot \frac{Q(\Lambda e^t,\ell+1)}{Q(\Lambda,\ell+1)} \right)^\prime~=~z-\Lambda e^t- (\ln Q(\Lambda e^t,\ell+1))^\prime\\
=~&z-\Lambda e^t+\frac{\Lambda e^t\cdot (Q(\Lambda e^t,\ell+1) - Q(\Lambda e^t,\ell))}{Q(\Lambda e^t,\ell+1)}.
\end{align*}
Substituting $T=\Lambda e^t$ we get
\begin{align*}
(zt-c(t))' &=z-T+ \frac{T\cdot\left( Q(T,\ell+1) - Q(T,\ell) \right)}{Q(T,\ell+1)} =z - T\cdot \frac{Q(T,\ell)}{Q(T,\ell+1)}.
\end{align*}
Setting this expression to zero and solving for $T$ gives the value of $T_z$ as in the statement of the lemma. The uniqueness of the solution for $z>\ell+1$ follows from the fact that the function $x\cdot \frac{Q(x,\ell)}{Q(x,\ell+1)}$ is 
strictly increasing with respect to $x$ (cf. Claim~\ref{clm:incXi}) and, as $x$ approaches $0$, it tends to $\ell+1$. 
Letting $t_z$ be such that $T_z = \Lambda e^{t_z}$, we obtain
$$-c(t_z) = -T_z- \ln Q(T_z,\ell+1) +\Lambda +\ln Q(\Lambda,\ell+1) $$
and 
$$t_z z=z(\ln T_z - \ln \Lambda).$$
The function $-c(t)$ is concave with respect to $t$ (cf.\ Proposition VII.1.1 in~\cite[p.\ 229]{b:e06}); also adding the linear term $zt$ does preserve concavity. So $t_z$ is the point where the unique maximum of $zt-c(t)$ is attained over $t \in {\mathbb R}$.  Combining the above we obtain $I_{\Lambda}(z)$ as stated in the lemma.
For $z=\frac{\Lambda Q(\Lambda,\ell)}{Q(\Lambda,\ell+1)}$ we have $ T_z= \Lambda$ which yields $I_\Lambda(\mu)=0$.
As far as $I_\Lambda(\ell+1)$ is concerned, 
note that strictly speaking this is not defined, as there is no positive solution of the equation $\ell+1=  T\cdot \frac{Q(T,\ell)}{Q(T,\ell+1)}$. 
We will express $I_\Lambda(\ell+1)$ as a limit as $T\rightarrow 0$ from the right and show that 
$$ \Pr{ \sum_{i=1}^s X_i\leq s(\la +1)}= \exp(-s I_\Lambda(\la +1)).$$
We define
\begin{align*}
I_\Lambda(\ell+1 ) & := \lim_{T\rightarrow 0^+}\left((\ell+1) \ln T -T- \ln Q(T,\ell+1)\right) -(\ell+1)\ln \Lambda+\Lambda + \ln Q(\Lambda,\ell+1).
\end{align*}
But 
\begin{align*} 
\lim_{T\rightarrow 0^+}\left((\ell+1) \ln T -T- \ln Q(T,\ell+1)\right)
&=\lim_{T \rightarrow 0^+} \ln { T^{\ell+1} \over e^TQ(T,\ell+1)}  \\
&=\lim_{T \rightarrow 0^+} \ln {T^{\ell+1} \over {T^{\ell+1}\over (\ell+1)!} + {T^{\ell+2}\over (\ell+2)!}+ \cdots}  \\
&=\lim_{T \rightarrow 0^+} \ln {1 \over {1 \over(\ell+1)!} + {T \over (\ell+2)!}+ \cdots} = \ln (\ell+1)!~,
\end{align*}
and therefore
$$ I_\Lambda(\ell+1)= \ln (\ell+1)! -(\ell+1)\ln \Lambda +\Lambda+ \ln Q(\Lambda,\ell+1).$$
On the other hand, the independence of the $X_i$'s guarantees that
 \begin{equation*}
\begin{split}
  \Pr{ \sum_{i=1}^s X_i\leq s(\la +1)}=& [\Pr{X_1=\la +1}]^s=\left({{e^{-\Lambda}\Lambda^{\la +1} \over (\la+1)!}\over Q(\Lambda, \la +1)}\right)^s=\exp(-s I_\Lambda(\la +1)).
\end{split}
 \end{equation*}
Also, according to Theorem~\ref{thm:largedeviation} the function $I_\Lambda(z)$ is non-negative and convex on its domain. 
So if $z\leq \mu$, then $\inf \{I_\Lambda(x) : x\leq z \} = I_\Lambda(z)$ and the second part of the lemma follows.
\end{proof}
Theorem II.3.3 in~\cite{b:e06} along with the above lemma then implies the following corollary.
\begin{corollary}
\label{cor:largedeviation}
 Let $X_1,\ldots,X_s$ be i.i.d.~$(\ell+1)-$truncated Poisson random variables with parameter $\Lambda$ and set $\mu= \ex(X_1)$. For any $\varepsilon>0$ there exists a constant $C=C(\varepsilon)>0$ such that as $s\rightarrow \infty$
$$\Pr { \bigg|{\sum_{i=1}^s X_i } - s\mu \bigg|\ge s\varepsilon} \le e^{-Cs}.$$
\end{corollary}

With the above results in hand we are ready to prove the following corollary about the density of the $(\la +1)$-core.
\begin{corollary}
\label{cor:edgesCore}
 Let $\tilde{N}_{\ell+1}$ and $\tilde{M}_{\ell+1}$ denote the number of vertices and edges in the $(\ell+1)$-core of $\widetilde{H}_{n,p,k}$. Also let $ck=p {n-1\choose k-1}$. Then, for any $0<\delta< 1$, with probability $1-n^{-\omega(1)}$,
\begin{align}
\tilde{N}_{\ell+1} & = Q(\xi,\ell+1)n\pm \delta n, \\
\tilde{M}_{\ell+1} & = \frac{\xi Q(\xi,\ell)}{kQ(\xi,\ell+1)}\tilde{N}_{\ell+1} \pm \delta n,
\end{align}
where $\xi:= \bar{x}ck$ and $\bar{x}$ is the largest solution of the equation $x = Q(xck,\ell)^{k-1}$. 
 \end{corollary}
\begin{proof}
The statement about $\tilde{N}_{\ell+1}$ follows immediately from the first part of Theorem~\ref{thm:CoreClone}. 

To see the second statement, we condition on certain values of $\tilde{N}_{\ell+1}$ and $\Lambda_{c,k,\ell}$ that lie in the intervals  stated in Theorem~\ref{thm:CoreClone}. In particular, we can assume that the total degree of the core of $\widetilde{H}_{n,p,k}$ is the sum of independent $(\ell+1)$-truncated Poisson random variables $d_1,\ldots, d_{\tilde{N}_{\ell+1}}$ with parameter $\Lambda_{c,k, \ell} =\xi + \beta$ for $|\beta|<\delta^2/2$. Let $D$ be the sum of the $d_i$'s. Therefore, Corollary~\ref{cor:largedeviation} yields for any $\varepsilon>0$ and a constant $C(\varepsilon)>0$ 
\[
\Pr {\left|D - \E {D}\right| \ge \varepsilon \tilde{N}_{\ell+1}} \le e^{-C(\varepsilon)\tilde{N}_{\ell +1}}. 
\]
The claim then follows from the fact that
\[ \E {D} ={ \Lambda_{c,k,\la} Q(\Lambda_{c,k,\la},\la)\over Q(\Lambda_{c,k,\la},\la+1)}\]
and the continuity of the above expression by choosing $\eps$ sufficiently small.
\end{proof}
We proceed with the proof of Theorem~\ref{thm:simple}, i.e., we will show that the $(\ell+1)$-core of $\widetilde{H}_{n,p,k}$ has density at least $\ell$ if $p = ck/\binom{n-1}{k-1}$ and $c > c_{k,\ell}^*$. Let $0 < \delta < 1$, and denote by $\tilde{N}_{\ell+1}$ and $\tilde{M}_{\ell+1}$ the number of vertices and edges in the $(\ell+1)$-core of $\widetilde{H}_{n,p,k}$. Applying Corollary~\ref{cor:edgesCore} we obtain that with probability $1 - n^{-\omega(1)}$
\begin{align*}
&\tilde{N}_{\ell+1}=Q(\xi,\ell+1)n\pm \delta n 
\quad\text{ and }\quad \\
&\tilde{M}_{\ell+1}=\frac{\xi Q(\xi,\ell)}{kQ(\xi,\ell+1)}\tilde{N}_{\ell+1} \pm \delta n,
\end{align*}
where $\xi = \bar{x}ck$ and $\bar{x}$ is the largest solution of the equation $x=Q(xck,\ell)^{k-1}$.
The value of~$c_{k,\ell}^*$ is then obtained by taking $\tilde{M}_{\ell+1} = \ell \tilde{N}_{\ell+1}$, and ignoring the additive error terms. The above values imply that the critical $\xi^*$ is given by the equation
\begin{equation}
\label{eq:xiast}
 \xi^* {Q(\xi^*,\ell)\over kQ(\xi^*,\ell+1)} =\ell
\implies
k\ell = {\xi^*{Q(\xi^{*},\ell)\over Q(\xi^*,\ell+1) } }.
\end{equation}
This is precisely~\eqref{eq:kxi}. So, the product $k\ell$ determines $\xi^*$ and $\bar{x}$ satisfies $\bar{x} = Q(\bar{x}ck,\ell)^{k-1}= Q(\xi^*,\ell)^{k-1}$. 
Therefore, the critical density is
\begin{equation} \label{eq:c_k}
c_{k,\ell}^*={\xi^* \over \bar{x} k}= {\xi^*\over k Q(\xi^*,\ell)^{k-1}}.
\end{equation}
\begin{proof}[Proof of Theorem~\ref{thm:simple}]
The above calculations imply that uniformly for any $0 < \delta < 1$, with probability $1-o(1)$
\[
	{\tilde{M}_{\ell+1} \over \tilde{N}_{\ell+1}} = {1\over k}~{\xi Q(\xi,\ell) \over Q(\xi,\ell+1)} \pm \Theta(\delta).
\]
In particular, if $c = c_{k,\ell}^*$, then ${\tilde{M}_{\ell+1}}/{\tilde{N}_{\ell+1}} = \ell \pm \Theta(\delta)$. To complete the proof it is therefore sufficient to show that the ratio ${\xi Q(\xi,\ell) \over Q(\xi,\ell+1)}$ is an increasing function of $c$. Note that this is the expected value of an $(\ell+1)$-truncated Poisson random variable with parameter $\xi$, which is  increasing in $\xi$ (cf. Corollary~\ref{clm:xiinc}). Recall that $\xi = \bar{x}ck$. We conclude the proof by showing the following claim.
\begin{claim} 
\label{cl:coreDensity}
The quantity $\xi=\bar{x}ck$ is increasing with respect to $c$. So, for some fixed $c$,  with probability $1-o(1)$
\begin{align*}
	&\frac{\tilde{M}_{\ell+1}}{\tilde{N}_{\ell+1}} < \ell ~~\text{, if $c < c_{k,\ell}^*$}
	\qquad \text{and}\qquad \frac{\tilde{M}_{\ell+1}}{\tilde{N}_{\ell+1}} > \ell ~~\text{, if $c > c_{k,\ell}^*$}.
\end{align*}
\end{claim}
\noindent 
Indeed, recall that $\bar{x}$ satisfies $\bar{x} = Q(\bar{x}ck,\ell)^{k-1}$. Equivalently, 
$\bar{x}ck = ck\cdot Q(\bar{x}ck,\ell)^{k-1}$. We have 
\begin{equation} \label{eq:Root} ck = {\xi \over Q(\xi,\ell)^{k-1}}.\end{equation}
The derivative of the function $F(\xi):={\xi \over Q(\xi,\ell)^{k-1}}$ with respect to $\xi$ is given by 
\[Q(\xi,\ell)^{-k}\left( Q(\xi,\ell)-(k-1)\xi \cdot \Pr{\mathrm{Po}(\xi)=\ell-1}\right).\] 
An easy calculation shows that $F'(\xi)$ is positive when $\xi$ satisfies the inequality
 \[\sum_{i\ge \la}{ {\xi}^{i-\la}\over i!}>{k\over (\la-1)!}, \]
 and negative otherwise. We therefore conclude that $F(\xi)$ is a convex function. Moreover, by the assumption in Theorem~\ref{thm:CoreClone} we have $ck > \min_{x>0} ({x}/{Q(x,\ell)^{k-1}})$. This implies the function ${\xi  \cdot Q(\xi,\ell)^{-(k-1)}}$ is strictly increasing in the domain of interest. Note that by~\eqref{eq:Root} the first derivative of $\xi$ with respect to $c$ is given by $k / F'(\xi)$ which is positive by the above discussion, thus proving our claim.
 
\end{proof}

\section{Proof of Theorem~\ref{thm:diff}}
\label{sec:proofMain}

Let us begin with introducing some notation. For a hypergraph~$H$ we will denote by~$V_H$ its vertex set and by~$E_H$ its set of edges. Additionally, we write $v_H=|V_H|$ and~$e_H=|V_H|$. For~$U \subset V_H$ we denote by~$v_U$,~$e_U$ the number of vertices in~$U$ and the number of edges joining vertices only in~$U$. Finally,~$d_U$ is the total degree in~$U$, i.e., the sum of the degrees in~$H$ of all vertices in~$U$. We say that a subset $U$ of the vertex set of a hypergraph is $\ell$-\emph{dense}, if $e_U/v_U \ge \ell$. By a  \emph{maximal} $\ell$-dense subset we mean that whenever we add a vertex to such a set, then its density drops below $\ell$.

In order to prove Theorem~\ref{thm:diff} we will to show that whenever $c<c^*_{k,\ell}$, the random graph $H_{n,\lfloor cn\rfloor,k}$ does not contain any $\ell$-dense subset with probability $1-o(1)$. We will accomplish this by proving that such a hypergraph does not contain any maximal $\ell$-dense subset with probability $1-o(1)$. Note that this is sufficient as any $\ell$-dense subset will be contained in some maximal $\ell$-dense subset. We shall use the following property.
\begin{proposition}
Let $H$ be a $k$-uniform hypergraph with density less than $\ell$ and let $U$ be a maximal $\ell$-dense subset of $V_H$. Then there is a $0\le \theta<\ell$ such that $e_U=\ell\cdot v_U +\theta$. Also, for each vertex $v\in V_H\setminus U$ the corresponding degree $d$ in $U$, i.e., the number of edges in $H$ that contain $v$ and all other vertices only from $U$, is less than $\ell-\theta$. \label{prop:Umaxldense}
\end{proposition}
\begin{proof}
If $\theta \geq \ell$, then we have $e_U \geq \ell  \cdot (v_U+1) $. Let $U^{\prime}= U \cup \{v\}$, where $v$ is any vertex in $V_H\setminus U$. Note that such a vertex always exists, as $U\neq V_H$. Let $d$ be the degree of $v$ in $U$. Then
$$\frac{e_{U^{\prime}}}{v_{U^{\prime}}}= \frac{e_U+d}{v_U+1} \geq \frac{e_U }{v_U+1}\geq \ell,$$
which contradicts the maximality of $U$ in $H$. Similarly, if there exists a vertex $v\in V_H\setminus U$ with degree $d\geq \ell-\theta$ in $U$, then we could obtain a larger $\ell$-dense subset of $V_H$ by adding $v$ to $U$.
\end{proof}
We begin with showing that whenever $c<\la$, the random graph $H_{n,cn,k}$ does not contain small maximal $\ell$-dense subsets. In particular, the following lemma argues about subsets of size at most $0.6n$. 
\begin{lemma} \label{lem:smallU1}
Let $ c<\la$ and $k\geq 3,~\ell\geq 2 $. With probability $1-o(1)$, $H_{n,\lfloor cn \rfloor,k}$ contains no maximal $\ell$-dense subset with less than $0.6n$ vertices.
\end{lemma}
\begin{proof}

We first prove the lemma for all $k \ge3$ and $\ell \ge 2$ except for the case $(k,\ell) \neq (3,2)$ by using a rough first moment argument. The probability that an edge of $H_{n,cn,k}$ is contained completely in a subset $U$ of the vertex set is given by \[{|U|\choose k}/{n \choose k}\leq \left(\frac{|U|}{n}\right)^k.\] Let $k/n\leq u\leq 0.6$ and for $x\in(0,1)$ let $H(x)=-x\ln x-(1-x)\ln (1-x)$ denote the entropy function. Then
\begin{equation} \label{eq:l-dense}
\begin{split}
\Pr{\exists\text{$\ell$-dense subset with $un$ vertices}} & \leq {n\choose un}\cdot {cn \choose \ell un}(u^k)^{\ell un} \leq e^{n((\ell+1)H(u)+k\ell u\ln u)}.
\end{split}
\end{equation}
We first show that the exponent attains its maximum at $u=k/n$ or $u=0.6$. Let $u_{max}= 1-{(\ell+1)/k\ell}$. We note that the second derivative of the exponent in \eqref{eq:l-dense} equals $${(k\ell(1-u)-(\ell+1))}/({u(1-u)}),$$ which is positive for $k\geq 3,\ell\geq 2$ and $u\in(0,u_{max}]$. Hence the exponent is convex for $u\leq u_{max}$, implying that it attains a global maximum at $u=k/n$ or at $u=(k\ell-(\ell+1))/k\ell$. Moreover, for any $k\geq 4,\ell \geq 2$ we have $u_{max}>0.6$. The case $k=3$ and $\ell \geq 3$ is slightly more involved. Note that  $u_{max}\geq 5/9$ in this case. The second derivative of the exponent is negative for $u\in (u_{max},1)$, implying that the function is concave in the specified range. But the first derivative of the exponent is $(\ell+1) \ln((1-u)/u) +3\ell(1+ \ln(u))$, which is at least $  2.8\ell - 0.41>0$ for $u=0.6$. Hence, the exponent is increasing at $u=0.6$. 

We can now infer that for $k= 3$,~$\ell\ge 3$ and $k\ge 4$,~$\ell\ge2$ , the exponent is either maximized at $u=k/n$ or at $u=0.6$. Note that
\begin{equation*}
\begin{split}
(\ell+1)H\left(\frac{k}{n}\right) +\frac{k^2\ell}{n}\ln\left(\frac{k}{n}\right)
=-\frac{(k^2\ell-(\ell+1)k)\ln n}{n} 
+O\left(\frac{1}{n}\right).
\end{split}
\end{equation*}
Also for $k\geq 4$ and $\ell\ge 2$ we obtain
\begin{align*}
(\ell+1) H(0.6) + k\ell\cdot 0.6 \ln(0.6)&\leq (\ell+1)H(0.6) + 4\ell  \cdot 0.6 \ln(0.6) \\
&\leq H(0.6)-0.56 \ell\leq -0.44,
\end{align*}
and for $k=3$ and $\ell\geq 3$
\begin{align*}
\begin{split}
(\ell+1) H(0.6) + k\ell\cdot 0.6 \ln(0.6) &\leq (\ell+1)H(0.6) + 3\ell  \cdot 0.6 \ln(0.6) \\
&\leq H(0.6)-0.24 \ell\leq -0.04.
\end{split}
\end{align*}
So, the maximum is obtained at $u=k/n$ for $n$ sufficiently large, and we conclude the case in which $(k,\ell) \neq (3,2)$ with
\begin{align*}
 \Pr{\exists~\ell\text{-dense subset with $\leq 0.6n$ vertices}}\le \sum_{u=k/n}^{0.6}n^{-k^2\ell +(\ell+1)k}
=O(n^{-8}).
\end{align*}

For the case $(k,\ell) = (3,2)$ a counting argument as above involving the $2$-dense sets does not work, and we will use the property that the considered set are \emph{maximal} 2-dense.
By~\eqref{eq:c_k} we obtain $c^*_{3,2}< 1.97$. Let $p = c^{\prime}/{n-1\choose 2}$, where $c^{\prime}= 3\cdot c\le 3\cdot c^*_{3,2}\le 5.91$. A simple application of Stirling's formula
reveals
$$\Pr {H_{n,p,3} \text{ has exactly $cn$ edges}} = (1 + o(1))(2\pi cn)^{-1/2}.$$
Let $U$ be a maximal $2$-dense subset of $H_{n,cn,3}$. 
As the distribution of $H_{n,cn,3}$ is the same as the distribution of $H_{n,p,3}$ conditioned on the
number of edges being precisely $cn$ we infer that
\begin{eqnarray*}
\Pr{H_{n,cn,3} \text{ contains a maximal $2$-dense subset $U$ with at most $0.6n$ vertices}}=\\ 
O(\sqrt{n})\cdot \Pr {H_{n,p,3} \text{ contains a maximal $2$-dense subset $U$ with at most $0.6n$ vertices}}.
\end{eqnarray*}
To complete the proof it is therefore sufficient to show that the latter probability is $o(n^{-1/2})$.
By Proposition~\ref{prop:Umaxldense} the event that $H_{n,p,3}$ contains  a maximal $2$-dense subset  $U$ implies that there exists a $\theta\in \{0,1\}$ such that $e_U =  2\cdot v_U +\theta$ and all vertices in $V_H\setminus U$  have degree less than $2-\theta$ in $U$. 
 We will show that the expected number of such sets with at most $0.6n$ vertices is $o(1)$.
We accomplish this in two steps. Note that if a subset $U$ is maximal $2$-dense, then certainly $|U|\geq 5$.
Let us begin with the case $s := |U|\leq n^{1/3}$. There are at most $n^s$ ways to choose the vertices in $U$, and at
most $s^{3(2s+\theta)}$ ways to choose the edges that are contained in $U$. Hence, for large $n$ the probability that $H_{n,p,3}$
contains such a subset with at most $\lfloor n^{1/3} \rfloor$ vertices is bounded by
\begin{align*}
 \sum_{s=5}^{\lfloor n^{1/3} \rfloor}\sum_{\theta=0}^1n^ss^{6s +3\theta}p^{2s+\theta}
< & \sum_{s=5}^{\lfloor n^{1/3} \rfloor}2n^ss^{6s +3}p^{2s}
=\sum_{s=5}^{\lfloor n^{1/3} \rfloor}2\left( ns^6 \left(\frac{c^{\prime}}{{n-1\choose 2}}\right)^2\right)^s\cdot s^3\\
\leq & ~n\sum_{s=5}^{\lfloor n^{1/3} \rfloor}2\left(c'^2n^{(1+6/3)-4}\right)^s 
\leq ~n\sum_{s=5}^{\lfloor n^{1/3} \rfloor}\left(n^{-1+o(1)}\right)^s=n^{-4+o(1)}.
\end{align*}
Let us now consider the case $n^{1/3}\leq |U|\leq 0.6n$. We note that 
\begin{equation*}
 \ln p=\ln  \left(\frac{c^{\prime}}{{n-1\choose 2}}\right)= \ln \frac{2c^{\prime}}{n^2} + \Theta\left({1 \over n}\right).
\end{equation*}
Also, there are ${n\choose un}\leq e^{nH(u)}$ ways to select $U$. Moreover, the number of ways to choose the $2un+\theta$ edges that are completely contained in $U$ is 
\begin{align*}
 {{un \choose 3}\choose 2un+\theta} \le \left(\frac{e (un)^3}{6(2un+\theta)}\right)^{2un}
=~\exp \left\lbrace 2un \ln \left(\frac{e (un)^2}{12}\right)+O(1)\right \rbrace.
\end{align*}
Finally, the probability that a vertex outside of  $U$ has a degree less than $2-\theta$ in $|U|$ is at most
\begin{equation*}
(1-p)^{un \choose 2} + {un \choose 2}p(1-p)^{{un \choose 2}-1}
=e^{-u^2c'}(1+u^2c')(1+O(1/n)).
\end{equation*}
Combining the above facts we obtain that the probability $P_u$ that $H_{n,p,3}$ contains a maximal $2$-dense subset $U$ with $2un$ vertices is
\[
\begin{split}
	P_u \le
	&\sum_{\theta=0}^1\binom{n}{un}\binom{\binom{un}{3}}{2un+\theta} p^{2un+\theta}(1-p)^{\binom{un}{3} - 2un-\theta} 
	\cdot \left(e^{-u^2c^{\prime}}(1+u^2c')(1+O(1/n))\right)^{(1-u)n}\\
	\le &~
	\exp \bigg \lbrace n\left( H(u) + 2u\ln \left(\frac{eu^2n^2}{12}\right) +2u\ln p\right) -p\left( {un\choose 3}-2un-1\right)~\\
&~~~~~~+ (1-u)n(-u^2c'+\ln(1+u^2c') )+O(1/n)\bigg\rbrace \\
\leq&\exp \bigg \lbrace n\bigg( H(u) + 2u\ln \left(\frac{ec'u^2}{6}\right)-\frac{u^3c'}{3}+(1-u)(-u^2c'+\ln(1+u^2c') ) \bigg)+O(1/n) \bigg \rbrace.
\end{split}
\]
If we fix $u$, the derivative of the exponent with respect to $c'$ is given by 

\begin{align*}
{2u\over c'}-{u^3\over 3}+(1-u)\left(-u^2+{u^2\over 1+u^2c'}\right)&\stackrel{c' \leq 5.91}{\ge}
{2u\over 6}-{u^3\over 3}+(1-u)\left(-u^2+{u^2\over 1+6u^2}\right) \\
=&u\left ( {1\over 3} -{{u^2/3}+6u^3-4u^4\over 1+6u^2}\right) 
\stackrel{u\le 0.6}\ge u\left ( {1\over 3} - 0.29 \right)  &\stackrel{u> 0} > 0,
\end{align*}
 thus implying that for all $u \in (0,0.6] $ the exponent is increasing with respect to $c'$. Therefore, it is sufficient to consider only the case when $c'=5.91$.  

The derivative of the exponent with respect to $u$ equals $\ln (c'^2u^3(1-u)) +6-\ln 6 -\ln(1+u^2c')-({(1-u)2u^3c'^2}/{(1+u^2c')})$. As the function $\ln(c'u^3)+(2u^4c'^3/ (1+u^2c'))$ is increasing and $\ln\left((1-u)/( 1+u^2c')\right)-(2u^3c'^2/(1+u^2c'))$ is decreasing in $u$, there is at most one $n^{-2/3}\leq u_0\leq 0.6$ where the derivative of the exponent vanishes. Moreover the derivative of the exponent at $u=0.6$ is positive. Therefore, $u_0$ is a global minimum, and the bound on $P_u$ is maximized at either at $u=n^{-2/3}$ or at $u=0.6$.  Elementary algebra then yields that the left point is the right choice, giving  the estimate $P_u = o(2^{-n^{1/3}})$, and the proof concludes by adding up this expression for all admissible $n^{-2/3} \le u \le 0.6$.

\end{proof}

%
%
%
%
In order to deal with larger subsets we switch to the Poisson cloning model. Let~$C$ denote the~$(\ell+1)$-core of~$\widetilde{H}_{n,p,k}$, where $p = ck/\binom{n-1}{k-1}$, and note that Theorem~\ref{cor:Contiguity} and Proposition~\ref{prop:ModelsEquiv} guarantee that $\widetilde{H}_{n,p,k}$ and $H_{n,cn,k}$ are sufficiently similar. Observe that any \emph{minimal}~$\ell$-dense set in~$\widetilde{H}_{n,p,k}$ is always a subset of~$C$, as otherwise, by removing vertices of degree at most $\ell$ the density would not decrease. In other words,~$C$ contains all minimal~$\ell$-dense subsets, and so it is enough to show that the core does not contain any~$\ell$-dense subset. Therefore, from now on we will restrict our attention to the study of~$C$.

Assume that the degree sequence of~$C$ is given by~$\mathbf{d} = (d_1,\ldots, d_{\tilde{N}_{\ell+1}})$, where we denote by~$\tilde{N}_{\ell+1}$ the number of vertices in~$C$. Thus, the number of edges in~$C$ is \[\tilde{M}_{\ell+1} = k^{-1}\sum_{i=1}^{\tilde{N}_{\ell+1}} d_i.\] For~$q,\beta \in[0,1]$ let~$X_{q,\beta}=X_{q,\beta}(C)=X_{q,\beta}(\mathbf{d})$ denote the number of subsets of~$C$ with~$\lfloor \beta \tilde{N}_{\ell+1}\rfloor$ vertices and total degree~$\lfloor qk\tilde{M}_{\ell+1}\rfloor$.

Let $\xi^\ast = \bar{x}^\ast c_{k,\ell}^\ast\, k$, where $\bar{x}^\ast$ is the largest solution of the equation $x=Q(xc^\ast_{k,\ell}k,\ell)^{k-1}$, and note that $\xi^\ast$ satisfies~\eqref{eq:xiast}. Moreover, let $\xi$ be given by $\xi= \bar{x}ck$, where $\bar{x}$ is the largest solution of the equation $x=Q(xck,\ell)^{k-1}$. As $\xi$ is increasing with respect to $c$ (cf.\ Claim~\ref{cl:coreDensity}), there exists a $\delta > 0$ and a $\gamma = \gamma(\delta)>0$ such that $c = c_{k, \ell}^*-\gamma$ and~$\xi = \xi^* - \delta$. Also $\gamma \rightarrow 0$ as $\delta \rightarrow 0$ by continuity of the largest solution of $x=Q(xck,\ell)^{k-1}$.

In the sequel we will assume that $\delta >0$ is fixed (and sufficiently small for all our estimates to hold), and we will choose $c < c_{k,\ell}^\ast$ such that $c = c_{k,\ell}^\ast - \gamma$ and $\xi = \xi^\ast - \delta$. Set
\begin{equation}
\label{eq:nm_l+1}
\begin{split}
n_{\ell+1}&= Q(\xi,\ell+1)n \quad\text{and}\quad m_{\ell+1}= \frac{\xi Q(\xi,\ell)}{kQ(\xi,\ell+1)}n_{\ell+1}.
	\end{split}
\end{equation}
By applying Corollary~\ref{cor:edgesCore} (and using $\delta^3$ instead of $\delta$) we obtain that with probability $1-n^{-\omega(1)}$
\begin{equation}
\label{eq:N2M2}
\tilde{N}_{\ell+1} = n_{\ell+1} \pm \delta^3 n ~\text{ and }~ \tilde{M}_{\ell+1} = m_{\ell+1} \pm \delta^3 n.
\end{equation}
Moreover, by applying Theorem~\ref{thm:CoreClone} we infer that $C$ is distributed like the cloning model with parameters $\tilde{N}_{\ell+1}$ and vertex degree distribution $\mathrm{Po}_{\ge \ell+1}(\Lambda_{c,k,\ell})$, where
\begin{equation}
\label{eq:Lambdackell}
\Lambda_{c,k,\ell}=\xi \pm \delta^3 = \xi^\ast - \delta \pm \delta^3,
\end{equation}
Recall that the definition of $\xi^\ast$ implies that $k\ell = \frac{\xi^\ast Q(\xi^\ast,\ell)}{Q(\xi^\ast,\ell+1)}$. Let $e_{k,\ell}$ denote the value of the first derivative of $\frac{x Q(x,\ell)}{k\ell Q(x,\ell+1)}$ with respect to $x$ at $x = \xi^\ast$. By
applying Taylor's Theorem to $\frac{x Q(x,\ell)}{Q(x,\ell+1)}$ around $x = \xi^*$ we obtain
\begin{equation} \label{eq:coresize}
\begin{split}
 m_{\ell+1}&= (1 - e_{k,\ell}\cdot \delta + \Theta(\delta^2))\ell  \cdot n_{\ell+1},
 \text{ where}\quad
 \frac{\xi Q(\xi,\ell)}{Q(\xi,\ell+1)}= k\ell(1 - e_{k,\ell}\cdot \delta + \Theta(\delta^2)).
 \end{split}
\end{equation}

Recall that $H_{\mathbf{d}, k}$ is a random hypergraph where the $i$th vertex has degree $d_i$.
We start by bounding the probability that a given subset of the vertices in $H_{\mathbf{d}, k}$ is maximal $\ell$--dense. In particular, we will work on the Stage 3 of the exposure process, i.e., when the number of vertices and degree sequence of the core have already been exposed. We will show the following.
\begin{lemma} \label{lem:BU}
Let $k \geq 3, \ell\geq 2$ and $\mathbf{d} = (d_1,\ldots, d_N)$ be a degree sequence and $U \subseteq\{1, . . . ,N\}$ such that $|U|=\lfloor \beta N \rfloor$. Moreover, set $M = k^{-1}\sum_{i=1}^N d_i$ and $q = (kM)^{-1}\sum_{i\in U}d_i$. Assume that $M < \ell  \cdot N$. If $\mathbb{P}_{\mathbf{d},k}$ denotes the probability measure on the space of $k$-uniform hypergraphs with degree sequence given by $\mathbf{d}$,  $\mathcal{B}(\beta,q)$ denotes the event that $U$ is a maximal $\ell$-dense set in $H_{\mathbf{d}, k}$,  and $H(x) = -x \ln x - (1 - x) \ln(1 - x)$ denotes the entropy function, then
$$P_{\mathbf{d},k}(\mathcal{B}(\beta,q))\leq O(M^{\ell+0.5})\binom{M}{\ell |U|}e^{-kMH(q)}(2^k-1)^{M-\ell |U|}.$$
\end{lemma}
\begin{proof}
Recall that~$H_{\mathbf{d}, k}$ is obtained by beginning with~$d_i$ clones for each~$1\le i\le N$ and by choosing uniformly at random a 
perfect~$k$-matching on this set of clones. This is equivalent to throwing~$kM$ balls into~$M$ bins such that every bin contains~$k$ balls. In
order to estimate the probability for~$\mathcal{B}(\beta, q)$ assume that we color the $kqM$ clones of the vertices in~$U$ with red, and the remaining $k(1-q)M$ clones with blue. Let $\theta$ be an integer such that $0\le \theta <\ell$. So, by applying
Proposition~\ref{prop:Umaxldense} we are interested in the probability for the event that there are exactly~$B_\theta=\ell|U|+\theta$ bins with~$k$ red balls.
We estimate the above probability as follows. We begin by putting into each bin $k$ \emph{black} balls, labeled with the numbers 
$1, \dots, k$. Let~$\mathcal{K} = \{1, \dots, k\}$, and let~$X_1, \dots, X_M$ be independent random sets such that for $1 \le i\le M$
\[
	\forall \mathcal{K}' \subseteq \mathcal{K} ~:~ \Pr{X_i = \mathcal{K}'} = q^{|\mathcal{K}'|}(1-q)^{k - |\mathcal{K}'|}.
\]
Note that $|X_i|$ follows the binomial distribution $\Bin(k, q)$. We then recolor the balls in the $i$th bin that are in $X_i$ with red, and all others with blue. 
 So, the total number of red balls is $X = \sum_{i=1}^M |X_i|$. Note that $\E{X} = kqM$, and that $X$ is distributed as $\Bin(kM, q)$. A straightforward application of Stirling's formula then gives
\[
	\Pr{X = kqM  } =\Pr{X = \E{X}} = (1 + o(1))(2\pi q(1-q)kM)^{-1/2}.
\]
 Let $R_j$ be the number of $X_i$'s that contain $j$ elements. Then 
\begin{equation}
\label{eq:probBU}
\begin{split}
	\mathbb{P}_{\mathbf{d},k}{(\mathcal{B}(\beta,q))} \le &\sum_{\theta=0}^{\la-1}\Pr{R_k = B_\theta |X = kqM} =\sum_{\theta=0}^{\la-1}{\Pr{X = kqM \wedge R_k = B_\theta}\over \Pr{X = kqM}}\\
=&~O\left(\sqrt{M}\right) \sum_{\theta=0}^{\la-1}\Pr{X = kqM \wedge R_k = B_\theta} .
\end{split}
\end{equation}
Let $p_j = \Pr{|X_i| = j} = \binom{k}{j}q^j(1-q)^{k-j}$. Moreover, define the set of integer sequences
\begin{align*}
\mathcal{A}=&\bigg\{(b_0, \dots, b_{k-1}) \in \mathbb{N}^{k}
:\sum_{j=0}^{k-1} b_j = M - B_\theta \textrm{ and } \sum_{j=0}^{k-1} j b_j = kqM - kB_\theta \bigg\}.
\end{align*}
Then
\begin{align*}
	 &\Pr{X = kqM \wedge R_k = B_\theta}
\le\sum_{\theta=0}^{\la-1} \sum_{(b_0, \dots, b_{k-1})\in \mathcal{A}} \binom{M}{b_0, \dots, b_{k-1}, B_\theta} \cdot 
\left( \prod_{j=0}^{k-1}p_j^{b_j} \right) \cdot p_k^{B_\theta}.
\end{align*}
Now observe that the summand can be rewritten as
\[
	\binom{M}{B_\theta} q^{kqM}(1-q)^{k(1-q)M} \cdot \binom{M-B_\theta}{b_0, \dots, b_{k-1}} \prod_{j=0}^{k-1} \binom{k}{j}^{b_j}.
\]
Also,
\begin{equation*}
 \sum_{(b_0, \dots, b_{k-1}) \in \mathcal{A}}\binom{M-B_\theta}{b_0, \dots, b_{k-1}} \prod_{j=0}^{k-1}
	\binom{k}{j}^{b_j}\le \left( \sum_{j=0}^{k-1} \binom{k}{j} \right)^{M-B_{\theta}}
	=(2^k-1)^{M-B_{\theta}}.
\end{equation*}
Thus, we have
\begin{equation*}
\begin{split}
	\Pr{X = kqM \wedge R_k = B_\theta}
	\le &\sum_{\theta=0}^{\la-1}\binom{M}{B_\theta} q^{kqM}(1-q)^{k(1-q)M} (2^k-1)^{M-B_\theta}\\
	\le &\sum_{\theta=0}^{\la-1}M^{\theta}\binom{M}{\ell |U|}e^{-kMH(q)}(2^k-1)^{M-\ell |U|} \cdot(2^k-1)^{-\theta}\\
	\le &~\la M^{\ell}\binom{M}{\ell |U|}(2^k-1)^{M-\ell |U|}e^{-kMH(q)} .
\end{split}
\end{equation*}
The claim then follows by combining the above facts and \eqref{eq:probBU}.
\end{proof}
As already mentioned, the above lemma gives us a bound on the probability that a subset of the $(\ell+1)$-core with a given number of
vertices and total degree is maximal $\ell$-dense, assuming that the degree sequence is given. In particular, we work on the probability space
of Stage 3 of the exposure process. In order to show that the $(\ell+1)$-core contains no $\ell$-dense subset, we will estimate the number of
such subsets. Recall that $X_{q, \beta}(\mathbf{d})$ denotes the number of subsets of~$H_{\mathbf{d}, k}$ with~$\lfloor \beta
\tilde{N}_{\ell +1} \rfloor$ vertices and total degree~$\lfloor q\cdot k \tilde{M}_{\ell +1} \rfloor$. Let also~$X_{q,\beta}^{(\ell)}$ denote
the number of these sets that are maximal~$\ell$-dense.  As an immediate consequence of Markov's inequality we obtain the following corollary.

\begin{corollary} \label{cor:existence}
Let $\B(q, \beta)$ be defined as in Lemma~\ref{lem:BU}, and let $\mathbf{d}$ be the degree sequence of the core of $\Pcl$. Then
$$ \Pr { X_{q,\beta}^{(\ell)} >0 ~|~ \mathbf{d}} \leq X_{q,\beta}(\mathbf{d}) \mathbb{P}_{\mathbf{d}, k} (\B(q, \beta)).$$
\end{corollary}
By applying Lemma~\ref{lem:smallU1} we obtain that~$H_{n,cn,k}$ does not obtain any $\ell$-dense set with less that~$0.6n$ vertices. This is particularly also true for $C$, and so it remains to prove Theorem~\ref{thm:diff} for sets of size bigger than~$0.6n \ge 0.6\tilde{N}_{\ell +1}$. We also observe that it is sufficient to argue about subsets of size up to, say,~$(1-e_{k,\ell}\delta/2)\tilde{N}_{\ell +1}$, as~\eqref{eq:coresize} implies that for small~$\delta$ all larger subsets have density smaller than~$\ell$. Moreover, the total degree $D$ of any~$\ell$-dense subset with $\beta \tilde{N}_{\ell +1}$ vertices is
 at least~$k\ell\cdot \beta \tilde{N}_{\ell+1}$, i.e.,
$$D=k\cdot q \tilde{M}_{\ell+1}  \Rightarrow k\ell\cdot \beta \tilde{N}_{\ell+1}\le k\cdot q \tilde{M}_{\ell+1}.$$
By~\eqref{eq:N2M2} and \eqref{eq:coresize}, we infer $\tilde{M}_{\ell+1} =\la(1-\Theta(\delta))$ which combined with above inequality implies that $q \ge (1+\Theta(\delta))\beta$.
Note that as each of the vertices in $C$ has degree at least $\ell +1$, the total degree of the~$(\ell +1)$-core with a $\ell$-dense subset with $\beta \tilde{N}_{\ell +1}$ vertices and degree $q \cdot k\tilde{M}_{\ell +1}$ satisfies
\begin{align*}
	k\tilde{M}_{\ell +1}&\ge q \cdot k\tilde{M}_{\ell +1} + (\ell +1)(\tilde{N}_{\ell+1} - \beta \tilde{N}_{\ell+1} )\\
&\Rightarrow q \le 1 - \frac{(\ell +1)(1-\beta)\tilde{N}_{\ell +1}}{k\tilde{M}_{\ell +1}}\stackrel{\eqref{eq:N2M2}, \eqref{eq:coresize}}{\le} 1 - \frac{(\ell +1)(1-\beta)}{k\ell},
\end{align*}
where the last inequality holds for any small enough $\delta$. Therefore, we fix $\beta$ and $q$ as follows.
\begin{equation}\label{eq:betaq}
 0.6<\beta<1-e_{k,\ell}\delta/2 \quad \text{and} \quad \la(1+\Theta(\delta))\beta \le q \le 1 - \frac{(\ell +1)(1-\beta)}{k\ell}.
 \end{equation}

\noindent
With Lemma~\ref{lem:BU} and Corollary~\ref{cor:existence} in hand we are ready to show the following.
\begin{lemma}\label{lem:Probs}
Let $m_{\ell+1}$ and $n_{\ell+1}$ be as defined in~\eqref{eq:nm_l+1} and $\CE$ be the event that \eqref{eq:N2M2} holds. Then
\begin{align*}
\Pr{X_{q,\beta}^{(\ell)}>0}= &~~\E{X_{q,\beta}|\mathcal{E}}(2^k-1)^{m_{\ell+1}-\ell\beta n_{\ell+1}}\cdot e^{\la n_{\la +1}H(\beta)-km_{\ell+1}H(q)+O(\delta^3 n)}+O\left( n^{-3}\right).
\end{align*}
\end{lemma}
\begin{proof} 
Let $\CE_1$ be the event that $X_{q,\beta}\leq n^{3} \ex (X_{q,\beta} \ | \ \CE )$. Markov's inequality immediately implies that $\Pr {\CE_1 \ | \ \CE} \geq 1 - n^{-3}$. 
If $\vec{d}$ is a vector, we write $\vec{d} \in \{\CE\cap \CE_1\}$ to denote that $\vec{d}$ is a possible 
degree sequence of $\C$ if the events $\CE$ and $\CE_1$ are realized.
We have 
\begin{equation*}
\begin{split}
\Pr {X_{q,\beta}^{(\ell)}> 0}
& \le ~\Pr {X_{q,\beta}^{(\ell)} >0 ~|~  \CE_1 \cap \CE} + \Pr {\overline{\CE_1}} + \Pr {\overline{\CE}} \\
& = \sum_{\vec{d} \in \{\CE\cap \CE_1\}}\Pr{X_{q,\beta}^{(\ell)} >0 ~|~ \CE_1 \cap \CE \text{ and } \mathbf{d} = \vec{d}}\cdot \Pr{\mathbf{d} = \vec{d}~|~ \CE_1 \cap \CE} + O(n^{-3}) \\
&=\sum_{\vec{d} \in \{\CE\cap \CE_1\}}\Pr{X_{q,\beta}^{(\ell)} >0 ~|~ \mathbf{d} = \vec{d}}\cdot \Pr{\mathbf{d} = \vec{d}~|~ \CE_1 \cap \CE} + O(n^{-3}) \\
&\stackrel{\text{Cor}.\ \ref{cor:existence}}=\sum_{\vec{d} \in \{\CE\cap \CE_1\}} X_{q, \beta}(\vec{d}) \mathbb{P}_{\vec{d},k}(\mathcal{B}(q, \beta))\cdot \Pr{\mathbf{d} = \vec{d}~|~ \CE_1 \cap \CE}+ O(n^{-3})\\
& = n^{3}~\E{X_{q, \beta} ~|~ \CE}\cdot \sum_{\vec{d} \in \{\CE\cap \CE_1\}} \mathbb{P}_{\vec{d},k}(\mathcal{B}(q, \beta))\Pr{\mathbf{d} = \vec{d}~|~ \CE_1 \cap \CE}+ O(n^{-3}).
\end{split}
\end{equation*}
Note that the assumption $\vec{d} \in \{\CE \cap \CE_1\}$ implies that the number of vertices $\tilde{N}_{\ell+1} $ of $\vec{d}$ is $n_{\ell +1} \pm \delta^3n$ and the number of edges $\tilde{M}_{\ell+1} $ is $m_{\ell +1}\pm\delta^3 n$, by $\CE$. Further note that for small enough $\delta$
\begin{align*} 
\tilde{M}_{\ell+1}
\le m_{\ell+1} + \delta^3 n
\le (1 - \Theta(\delta))\ell n_{\ell+1} + \delta^3 n
\le \ell \tilde{N}_{\ell+1} - \Theta(\delta)n
\end{align*}
Using Stirling's formula we obtain
 $$\binom{\tilde{M}_{\ell+1} }{\la \beta \tilde{N}_{\ell+1} } < \binom{\la \tilde{N}_{\ell+1} }{\la \beta \tilde{N}_{\ell+1} }= \exp(\ell n_{\ell+1} H(\beta) + O(\delta^3 n)) .$$
Thus, applying Lemma~\ref{lem:BU} we obtain uniformly for all $\vec{d} \in \{\CE \cap \CE_1\}$ that
\begin{align*}
	\mathbb{P}_{\bar{d},k}(\mathcal{B}(q, \beta)) =&(2^k-1)^{m_{\ell +1} - \beta n_{\ell +1}} \cdot e^{\la n_{\la+1}H(\beta)-km_{\ell +1}H(q) + O(\delta^3 n)}.
\end{align*}
The claim follows.
\end{proof}
The following lemma bounds the expected value of $X_{q,\beta}$ conditional on $\CE$.
\begin{lemma}\label{lem:ExpXdt}
There exists $\delta_0>0$ such that whenever $\delta < \delta_0$ 
\begin{eqnarray*}
\E{X_{q,\beta}|\mathcal{E}} <\exp \bigg(n_{\ell+1}H(\beta) -n_{\ell+1}(1-\beta)I_{\xi^*}\left( \frac{k\ell (1-q)}{1-\beta}\right) +  0.4 \cdot{ k\la \over \xi^*}\cdot n_{\la+1} \delta   + O(\delta^2 n)\bigg),
\end{eqnarray*}
where $I_{\xi^*}\left( \frac{k\ell (1-q)}{1-\beta}\right)$ is  the rate function  as defined in \eqref{Iz}.
\end{lemma}

\begin{proof}
 Let $t = \lfloor\beta \tilde{N}_{\ell+1}\rfloor$. Conditional on $\CE$ there are ${\tilde{N}_{\ell+1} \choose t} = e^{n_{\ell+1}H(\beta) +
O(\delta^3 n)}$ ways to select a set with~$t$ vertices. We shall next calculate the probability that one of them has the claimed property, and
the statement will follow from the linearity of expectation. Let~$U$ be a fixed subset of the vertex set of~$\C$ that has size~$t$.  We label
the vertices as $1, \ldots,\tilde{N}_{\ell+1}$ so that the vertices which are not in $U$ are indexed from $t+1$ to $\tilde{N}_{\ell+1}$. Let the random
variable $d_i$ denote the degree of vertex $i$.  We recall that $d_1,d_2,\ldots, d_{\tilde{N}_{\ell+1}} $ are i.i.d.\ $(\ell +1)$-truncated
Poisson variables with parameter~$\Lambda = \Lambda_{c,k,\la}=\xi \pm \delta^3$ and mean $\mu_\Lambda=   \Lambda {Q(\Lambda,\la) \over 
Q(\Lambda,\la +1)}.$ By Taylor's expansion of $\mu_{\lambda}$ around $\xi$ we obtain 
\[\mu_{\Lambda} = \xi{Q(\xi,\ell)\over Q(\xi,\ell +1)} \pm \Theta(\delta^3).\]

We will calculate the probability of the event $\sum_{i=1}^t d_i=q k\tilde{M}_{\ell+1}$ conditional on $\mathcal{E}$.  This is equivalent to
 calculating the probability of the event  $\sum_{i=t}^ {\tilde{N}_{\ell+1}}d_i=k(1-q)\tilde{M}_{\ell+1}$ conditional on $\mathcal{E}$ which by using \eqref{eq:nm_l+1} is same as the event 
 \[\sum_{i=t+1}^ {\tilde{N}_{\ell+1}}{d_i\over{\tilde{N}_{\ell+1}-t}} = \xi {Q(\xi,\la) \over Q(\xi,\la+1) } \cdot  {1-q \over 1-\beta} \pm \Theta(\delta^3).\]
Let us abbreviate $z=  \xi {Q(\xi,\la) \over Q(\xi,\la+1) } \cdot  {1-q \over 1-\beta} \pm \Theta(\delta^3)$. 
Using the lower bound of $q$ from~\eqref{eq:betaq} we obtain \[
z-\mu_{\Lambda} = \xi {Q(\xi,\la) \over Q(\xi,\la+1)}\cdot {\beta \over 1-\beta} \Theta(\delta) \pm  \Theta(\delta^3) >0.\]
As $I_\Lambda(x)$ is a non-negative convex function and $I_\Lambda(\mu_\Lambda)=0$,  $I_{\Lambda}(x)$ is a decreasing function for 
$x <\mu_\Lambda$. Therefore, by Lemma~\ref{lem:I} 
\[\Pr{\sum_{i=t+1}^ {\tilde{N}_{\ell+1}}{d_i}=z(\tilde{N}_{\ell+1}-t)~|~ \CE}= \exp {(-n_{\la+1}(1-\beta)\cdot I_\Lambda(z) (1+o(1)) )} \]
and\[ I_{\Lambda}(z)=z(\ln T_z-\ln \Lambda)- T_z + \Lambda- \ln Q(T_z,\ell+1)+\ln Q(\Lambda,\ell+1), \]
 where $T_z$ is the unique solution of $z= T_z \cdot {Q(T_z,\la) \over Q(T_z,\la+1)}$.
Note that 
\[ {\partial{I_{\Lambda}(z) }\over  \partial {\Lambda} }= - {z\over \Lambda} +1 + {{e^{-\Lambda}\Lambda^\la \over \la! }\over Q(\Lambda,  \la+1)}  = - {z\over \Lambda} + {Q(\Lambda, \la)\over Q(\Lambda, \la+1)} =  {\mu_\Lambda -z \over \Lambda}.\]
But recall that $\Lambda = \xi \pm \delta^3 =\xi^* -\delta \pm \delta^3$. So using Taylor's expansion around $\xi^*$ to write $I_\Lambda(z)$ in terms of 
$I_{\xi^*} (z)$  we obtain
\begin{align*}
I_\Lambda(z) =  &I_{\xi^*}(z) -\left(  {\mu_{\xi^*} -z \over \xi^*}\right) (\delta \pm \delta^3) \pm O(\delta^2)
 = I_{\xi^*}(z) - {\mu_{\xi^*} \over \xi^* } \cdot {q-\beta \over 1-\beta}~\delta \pm O(\delta^2).
 \end{align*}
The last equality holds as $z= \mu_{\xi^*} {1-q\over 1-\beta} (1 - e_{k,\ell} \delta + \Theta(\delta^2))$. 
Since $\beta > 0.6$ we have $q-\beta <0.4$. Also $\mu_{\xi^*}=k\la$.
 Therefore,
 \begin{equation} \label{eq:ILambda}
  I_\Lambda(z) \geq I_{\xi^*}(z) - {k\la\over \xi^*}\cdot{0.4 \over 1-\beta}~\delta  - \pm O(\delta^2).
  \end{equation}
We will now approximate $I_{\xi^*}(z)$ in terms of $I_{\xi^*}\left(k\la {1-q\over 1-\beta}\right).$ Note that
 \[ {\partial{I_{\xi^*}(z) }\over  \partial {z} }= \ln T_z -\ln \xi^*.\]
By Taylor's expansion of $I_\xi^*(z)$ around $z_0:=k\la{1-q\over 1-\beta}$ we obtain 
\begin{equation} \label{eq:TaylorII}
I_{\xi^*}(z) = I_\xi^*\left(k\la{1-q\over 1-\beta}\right)+\delta \cdot e_{k,\la} \left(k\la{1-q\over 1-\beta}\right)\left(\ln {\xi^*\over
T_{z_0}} \right)\pm O(\delta^2).\end{equation}
By Claim~\ref{clm:incXi} the function $\mu_{t} $ is increasing with respect to $t$. This implies that $T_{z_0} < \xi^*$ as 
$z_0<k\la$, whereby $\ln {\xi^*\over T_{z_0}} > 0$. Also recall that $e_{k,\ell}$ denotes the value of the partial derivative of 
${1\over k\la}\cdot \frac{tQ(t,\ell)}{ Q(t,\ell+1)}$ with respect to $t$ at $t = \xi^\ast$. Again, Claim~\ref{clm:incXi} implies that this 
is positive.  We  therefore obtain
\begin{equation} \label{eq:IXi}I_{\xi^*}(z) > I_\xi^*\left(k\la{1-q\over 1-\beta}\right) - \Theta (\delta^2) \end{equation}
Combining \eqref{eq:ILambda}, (\ref{eq:TaylorII}) and  \eqref{eq:IXi} we obtain
\[I_\Lambda(z) > I_\xi^*\left(k\la{1-q\over 1-\beta}\right) - {k\la\over \xi^*}\cdot{0.4 \over 1-\beta}~\delta - O(\delta^2) .\]
The proof is then completed by using the fact that $\Pr{\mathcal{E}} = 1 - n^{-\omega(1)}.$ 
\end{proof}
Lemma~\ref{lem:Probs} along with Lemmas~\ref{lem:BU} and~\ref{lem:ExpXdt} yield the following estimate.
\begin{lemma}\label{lem:final_prob}
There exists $\delta_0>0$ such that whenever $\delta < \delta_0$
\begin{align*}
\begin{split}
\Pr{X_{q,\beta}^{(\ell )} >0}<~& O(n^{-3}) + F(\beta,q;\ell), 
\end{split}
\end{align*}
where
\begin{align*}
\begin{split}
F(\beta,q;\ell) = ~&(2^k-1)^{m_{\ell+1}-\ell\beta n_{\ell+1}}\\
&~~~~~~~\cdot \exp\bigg((\la +1)n_{\ell+1}H(\beta)-km_{\ell+1}H(q) -n_{\ell+1}(1-\beta)I_{\xi^*}\left( \frac{k\ell(1-q)}{1-\beta}\right) \\
&~~~~~~~~~~~~~~~~~~~~~~~~~~~~~~~ +0.4 \cdot{ k\la \over \xi^*}\cdot n_{\la+1}\cdot \delta + O(\delta^2 n)\bigg),
\end{split}
\end{align*}
\end{lemma}
We can now complete the proof of Lemma~\ref{lem:Difficult} by showing the above probability is $o(1)$. We proceed as follows. Let
us abbreviate
\begin{align*}
f(\beta, q) :=&~
(\ell+1)H(\beta) + \ell  \cdot(1-\beta)\ln (2^k  -1) 
-k \ell \cdot H\left(q\right) - (1-\beta) I_{\xi^*}\left({k\ell(1-q) \over 1-\beta}\right).
\end{align*} 
By using Lemma~\ref{lem:final_prob} we infer that
\begin{align*} 
	\frac1{n_{\ell+1}}\ln F(\beta,q;\ell) \le &~f(\beta, q)+ e_{k,\ell}\cdot \delta \cdot k\ell \left(H\left(q\right)- { \ln(2^k-1)\over k} + {0.4 \over e_{k,\la }\cdot \xi^*} \right)  +O(\delta^2).
\end{align*}
By Claim~\ref{clm:ekl} $e_{k,\la } > 0.77/ \xi^*$. So
\begin{align} \label{eq:FinalBound}
	\frac1{n_{\ell+1}}\ln F(\beta,q;\ell) \le &~f(\beta, q)+ e_{k,\ell}\cdot \delta \cdot k\ell \left(H\left(q\right)- { \ln(2^k-1)\over k} + 0.52 \right)  +O(\delta^2).
\end{align}
We will now prove the main tool for the proof of Theorem~\ref{thm:diff}.
\begin{lemma} \label{lem:Difficult}
There exists $\hat{\delta} = \hat{\delta} (k,\ell)>0$ such that if $\delta <\hat{\delta}$ the following holds. 
With probability~$1-n^{-\omega(1)}$, for any~$0.6 < \beta \leq 1-e_{k,\ell}\delta/2$ and 
$\beta< q \leq 1 - \frac{(\ell +1)(1-\beta)}{k\ell}$, we have~$X_{q,\beta}^{(\ell)} = 0$. 
\end{lemma}
\begin{proof}
To deduce this lemma, we first bound $f(\beta, q)$.
\begin{claim}
\label{cl:techn}
For any $k\geq 3$ and $\ell \geq 2$, there exist $\eps_0, C > 0$ such that for any $\eps < \eps_0$ 
the following holds. For any $0.6 < \beta \le 1-\eps$, and $q$ as in 
Lemma~\ref{lem:Difficult}, we have 
\[
	f(\beta, q) \le -C \eps.
\]
\end{claim}
The proof of Lemma~\ref{lem:Difficult} will be complete as long as we show that for $\delta$ small enough the rest of the right-hand side 
of (\ref{eq:FinalBound}) is negative. Firstly, let $\delta_1 = \delta_1 (k,\ell)$ be such that for any $\delta < \delta_1$ we have 
$1- e_{k,\ell} \delta /2 > 0.999$. We will consider a case distinction according to the value of $q$.
 
If $q < 0.99$, then $\beta < 0.99$ as well, and Claim~\ref{cl:techn} implies that $f(\beta , q) \le - 0.01 \cdot C$, where $C>0$ depends on $k$ 
and $\ell$. Then let $\delta_2 =\delta_2 (k,\ell) >0$ be such that for $\delta < \delta_2$, we have 
$$e_{k,\ell}\cdot \delta \cdot k\ell \left(H\left( 0.6 \right)- { \ln(2^k-1)\over k} + 0.52 \right) + O(\delta^2) < 0.005 \cdot C.$$  
Here recall that $\beta \geq 0.6$.  So for any $\delta < \min \{\delta_0, \delta_1, \delta_2 \}$, (\ref{eq:FinalBound}) implies that 
$$ \frac1{n_{\ell+1}}\ln F(\beta,q;\ell) \le - 0.005 \cdot C.$$
Assume now that $q \geq 0.99$. The monotonicity of the entropy function implies that
\[
	H\left(q\right) - {\ln(2^k-1)\over k}+ 0.52
	\le H(0.99) - {\ln(2^k-1) \over k} +0.52 \stackrel{k\geq 3}< -0.072.
\]
Now with $0.6 \le \beta \le 1 - e_{k,\la}\cdot \delta/2$ as in Lemma~\ref{lem:Difficult}, the bound of Claim~\ref{cl:techn} substituted in 
(\ref{eq:FinalBound}) yields
\[
	\frac1{n_{\ell+1}}\ln F(\beta,q;\ell ) \le -C e_{k,\la}\cdot \delta/2 + O(\delta^2).
\]
In turn, this is at most  $-C e_{k,\la}\cdot \delta/4$, if $\delta < \delta_3 = \delta_3 (k,\ell)$. The above cases imply that if $\delta < \min \{\delta_0,\delta_1, \delta_2, \delta_3 \} =:\hat{\delta}$, then
with probability $1-e^{-\Omega( n_{\la+1})}-O(n^{-3})$ we have $X_{q,\beta}^{(\ell)} = 0$, for all $\beta$ and $q$ as 
in Lemma~\ref{lem:Difficult}.
\end{proof}

With the above result at hand we can finally complete the proof of Theorem~\ref{thm:diff}.
\begin{proof}[Proof of Theorem~\ref{thm:diff}]
Firstly, note that it is enough to argue that with probability $1 - o(1)$ the $(\ell+1)$-core does not contain any maximal $\ell$-dense subset; this follows from the discussion after Lemma~\ref{lem:smallU1}, which we do not repeat here. Moreover, by Theorem~\ref{cor:Contiguity} and Proposition~\ref{prop:ModelsEquiv}, it is enough to consider the~$(\ell+1)$-core~$C$ of~$\widetilde{H}_{n,p,k}$, where $p = ck/\binom{n-1}{k-1}$.

The proof is completed by applying Lemma~\ref{lem:Difficult}, as we can choose~$\delta > 0$ as small as we please.
\end{proof}
The rest of the paper is devoted to the proof of Claim~\ref{cl:techn} and contains a detailed analysis of the function $f$. We proceed as follows. We will fix arbitrarily a $\beta$ and 
we will consider $f(\beta,q)$ solely as a function of $q$. Then we will show that if $q_0 = q_0 (\beta)$ is a point 
where the partial derivative of $f$ with respect to $\beta$ vanishes, then $f(\beta, q_0) \le -C_1\eps$. 
Additionally, we will show that this holds for $f(\beta, \beta)$ and $f\left(\beta, 1 - \frac{(\ell +1)(1-\beta)}{k\ell}\right)$.  

\subsubsection*{Bounding $f(\beta,q)$ at its critical points} 

Let $\beta$ be fixed. We will evaluate $f(\beta,q)$ at a point where the partial derivative with respect to $q$ vanishes. 
To calculate the partial derivative with respect to $q$, we first need to determine the derivative of $I(z)$ with respect to $z$. 
According to Lemma~\ref{lem:I}, 
 $I_{\xi^*}(z) = z\left(\ln T_z - \ln \xi^* \right) - \ln Q(T_z,\ell+1)-T_z + \ln Q(\xi^*,\ell+1) +\xi^*$, where $T_z$ is the unique solution of $z= T_z\cdot \frac{Q(T_z,\ell)}{Q(T_z,\ell+1)}$.
Differentiating this with respect to $z$ we obtain
\begin{equation} \label{eq:IDer}
\begin{split}
I'_{\xi^*}(z)  =&\ln T_z - \ln \xi^*+  {z\over T_z}~{d T_z\over dz} -{d T_z\over dz}
- {Q(T_z,\ell)-Q(T_z,\ell+1)\over Q(T_z,\ell+1)}~{d T_z\over dz}\\
=&\ln T_z - \ln \xi^* +  {z\over T_z}~{d T_z\over dz}-{Q(T_z,\ell)\over Q(T_z,\ell+1)}~{d T_z\over dz}\\
=&\ln T_z - \ln \xi^*. 
\end{split}
\end{equation} 
However, in the differentiation of $f$ we need to differentiate $I_{\xi^*}( k\ell(1-q) / (1-\beta))$ with respect to $q$. Using 
(\ref{eq:IDer}), we obtain
$${\partial I_{\xi^*} \left({k\ell(1-q)\over 1-\beta} \right) \over \partial q} = -{k \ell \over 1-\beta}~\left(\ln H_q - \ln \xi^* \right),$$
where $H_q$ is the unique solution of the equation 
\begin{equation*} 
 {k\ell(1-q) \over 1-\beta} = {H_q \cdot Q(H_q,\ell) \over Q(H_q,\ell+1)}.
\end{equation*}
Observe that the choice of the range of $q$ is such that the left-hand side of the above equation is at least $\ell+1$. So, $H_q$ is well-defined.  Also, an elementary calculation shows that the derivative of the entropy function, $H'(q)$ is given by  $\ln \left( 1-q \over q \right)$. All the above facts together yield the derivative of $f(\beta,q)$ with respect to~$q$
\begin{equation*}
{\partial f(\beta, q)\over \partial q} = k\ell\left( -\ln \left({1-q\over q} \right) + 
\ln {H_q \over  \xi^*}  \right).
\end{equation*}
Therefore, if $q_0$ is a critical point, that is, if $\left. {\partial f(\beta, q)\over \partial q}\right|_{q=q_0} =0$, then with $T_0 = H_{q_0}$, $q_0$ satisfies 
\begin{equation}\label{eq:z_0Func}
T_0 = \xi^* \frac{1-q_0}{q_0} \quad and \quad {k\ell(1-q_0) \over 1-\beta} = {T_0 Q(T_0,\ell) \over Q(T_0,\ell+1)}.
\end{equation}
At this point, we have the main tool that will allow us to evaluate $f(\beta, q_0)$. 
We will use (\ref{eq:z_0Func}) in order to eliminate $T_0$ and express $f(\beta, q_0)$ solely as a function of $q_0$. 
\begin{claim} \label{claim:fCrit}
For any given $\beta \in (0.6, 1)$, if $q_0 = q_0(\beta)$ satisfies \eqref{eq:z_0Func}, then 
\begin{equation}
\label{eq:fCrit}
\begin{split}
f(\beta,q_0)&=\ln \bigg(e^{(\ell+1)H(\beta)} q_0^{k\ell}\left({(2^k -1)\left( 1 - q_0 \right)\over q_0}\right)^{\ell(1-\beta)}\cdot\left({(1-\beta)(k\ell-\xi^*) \over k\ell q_0-\xi^*(1-\beta)}\right)^{1-\beta}\bigg).
\end{split}
\end{equation}
\end{claim}
\begin{proof}
Note that 
\begin{equation*}
\begin{split}
I \left({k\ell(1-q_0) \over 1-\beta}\right)  = &~{k\ell(1-q_0)  \over 1-\beta} \ln {T_0\over  \xi^*} 
+ \ln \left(e^{\xi^*}Q(\xi^*,\ell+1)\over e^{T_0}Q(T_0,\ell+1)\right)  \\ 
\stackrel{(\ref{eq:z_0Func})}{=}&
{k\ell(1-q_0)\over 1-\beta} \ln \left({1-q_0 \over q_0} \right)+ \ln \left(e^{\xi^*}Q(\xi^*,\ell+1)\over e^{T_0}Q(T_0,\ell+1)\right).
\end{split}
\end{equation*}
Therefore,
\begin{equation*}
\begin{split}
-(1-\beta) I\left({k\ell(1-q_0) \over 1-\beta}\right) = &
-k\ell(1-q_0)\ln \left({1-q_0 \over q_0} \right) + (1-\beta) \ln \left(\frac{e^{T_0}Q(T_0,\ell+1)}{e^{\xi^*}Q(\xi,\ell+1)} \right) \\
=&- k\ell(1-q_0) \ln \left(1 - q_0 \right)  + k \ell\ln \left(q_0 \right)- k \ell q_0 \ln \left(q_0 \right) \\
&+ (1-\beta) \ln \left(\frac{e^{T_0}Q(T_0,\ell+1)}{e^{\xi^*}Q(\xi,\ell+1)} \right). 
\end{split}
\end{equation*}
Also, the definition of the entropy function implies that
\begin{equation*}
\begin{split}
-k\ell H\left( q_0 \right) &= k \ell q_0 \ln \left(q_0 \right) 
+ k\ell(1-q_0) \ln \left(1-q_0 \right).
\end{split}
\end{equation*}
Thus 
\begin{equation} \label{eq:fInter}
\begin{split}
-(1-\beta) I \left({k\ell(1-q_0) \over 1-\beta}\right) -k \ell H\left(q_0 \right)= \ln \left(q_0^{ k \ell} \left(\frac{e^{T_0}Q(T_0,\ell+1)}{e^{\xi}Q(\xi^*,\ell+1)} \right)^{1-\beta}\right).
\end{split}
\end{equation}
Let $z_0:={k\ell(1-q_0) \over 1 -\beta}$. 
Now we will express $e^{T_0}Q(T_0,\ell+1)$ as a rational function of $T_0$ and $z_0$. Solving 
(\ref{eq:z_0Func}) with respect to $e^{T_0}Q(T_0,\ell+1)$ yields 
\begin{equation*}
\begin{split}
e^{T_0}Q(T_0,\ell+1)&=e^{T_0}{ T_0 Q(T_0,\ell)\over z_0}={e^{T_0} T_0 \over z_0}\left(Q(T_0,\ell+1)+ e^{-T_0}{{T_0}^{\ell} \over \ell!} \right).
\end{split}
 \end{equation*}
Therefore,
\begin{equation*}
\begin{split}
e^{T_0}Q(T_0,\ell+1)&= {{T_0}^{\ell} \over \ell!}\left( {z_0\over T_0}-1 \right)^{-1}.
\end{split}
\end{equation*}
Note that 
\begin{equation*}
\begin{split}
z_0 - T_0 =&~~{k\ell (1-q_0) \over 1-\beta} - {\xi^* (1-q_0) \over q_0} 
=~~{(1-q_0) (k\ell q_0 - \xi^*(1-\beta ) ) \over (1-\beta)q_0}.
\end{split}
\end{equation*}
Thus we obtain 
\begin{equation*}
\begin{split}
\ln (e^{T_0} Q(T_0,\ell+1))  = & \ln\left({{T_0}^{\ell+1} \over (z-T_0)\ell !}\right) \\
 \stackrel{(\ref{eq:z_0Func})}{=}&  \ln \left(\bigg({\xi^*(1-q_0)\over q_0}\right)^{\ell+1}\cdot{(1-\beta)q_0\over (1-q_0) (k\ell q_0 - \xi^*(1-\beta ) )\ell!}  \bigg)\\
=& \ln \bigg({({\xi^*})^{\ell+1}\over \ell !}\left({1-q_0\over q_0}\right)^{\ell}\cdot {1-\beta\over  k\ell q_0 - \xi^*(1-\beta ) }  \bigg).
 \end{split}
\end{equation*}
Also, by definition of $\xi^*$ we have $k={\xi^* Q(\xi^*,\ell) \over \ell Q(\xi^*,\ell+1)}$ which is equivalent to $k\ell=\xi^*\left(1+{e^{-\xi^*}({\xi^*})^{\ell}/\ell !\over Q(\xi^*,\ell+1)}\right)$ and implies $e^{\xi^*}Q(\xi^*,\ell+1)={({\xi^*})^{\ell+1}/\ell!\over k\ell-\xi^*}$.
Substituting this into (\ref{eq:fInter}) and adding the remaining terms, we obtain (\ref{eq:fCrit}).
\end{proof}
We will now treat~$q_0$ as a free variable lying in the interval where~$q$ lies into, and we will study $f(\beta, q_0)$ for a fixed~$\beta$ as a function of~$q_0$. In particular, we will show that for any fixed~$\beta$ in the domain of interest~$f(\beta, q_0)$ is increasing. Thereafter, we will evaluate~$f(\beta, q_0)$ at the largest possible value that~$q_0$ can take, which is~$1 - {(\ell+1)(1-\beta )\over k\ell}$, and show that this value is negative. 
\begin{claim}
For any $k \geq 3,\ell\geq 2$ and for any $\beta > 0.6$ we have 
$${\partial f(\beta, q_0)\over \partial q_0} >0.$$
\end{claim}
\begin{proof}
The partial derivative of $f(\beta,q_0)$ with respect to $q_0$ is 
\begin{equation*}
\begin{split} 
{\partial f(\beta,q_0) \over \partial q_0} & = {k\ell \over q_0} - \ell{1-\beta \over 1-q_0} - \ell{ 1-\beta \over q_0} 
- {k\ell(1-\beta) \over k\ell q_0 - \xi^* (1-\beta)}.
\end{split}
\end{equation*}
Since $q_0 \leq 1-{(\ell+1)(1- \beta) \over k\ell}$, we obtain 
$$ 1- q_0 \geq {(\ell+1)(1-\beta) \over k\ell} \ \Rightarrow  - {1-\beta \over 1-q_0} \geq -{k\ell\over \ell+1}. $$ 
Also $q_0 \geq \beta$ and $\xi < k\ell$. Therefore, 
$$ k\ell q_0  -\xi (1-\beta)> k\ell \beta - k\ell(1-\beta) = 2\beta k\ell - k\ell = k\ell(2\beta -1).$$
Substituting these bounds into ${\partial f(\beta,q_0) \over \partial q_0}$ yields
\begin{equation*}
\begin{split}
{\partial f(\beta,q_0) \over \partial q_0} & > 
{k\ell\over q_0} - {k\ell^2 \over \ell+1} - {\ell(1-\beta) \over q_0} - {1-\beta \over 2\beta -1 } =  ~  {k\ell-\ell(1-\beta) \over q_0} -{k\ell^2 \over \ell+1} - {1-\beta \over 2\beta -1} \\
&\geq k\ell{k\ell-\la(1-\beta) \over k\ell-(\ell+1)(1-\beta)} -{k\ell^2 \over \ell+1} - {1-\beta \over 2\beta -1} 
\geq k\left(\ell-{\ell^2 \over \ell+1} - {1-\beta \over k(2\beta -1)}\right)\\
&= k\left({\ell \over \ell+1} -{1-\beta \over k(2\beta -1)}\right).
\end{split}
\end{equation*}
But 
$${\ell\over \ell+1} > {1-\beta \over k(2\beta -1)},$$
as  $k\ell(2\beta -1) > (\ell+1)(1-\beta)$, which is equivalent to $\beta > {(k\ell+\ell+1) / (2k\ell+\ell+1)}$. Elementary algebra then yields
that ${(k\ell+\ell+1) / (2k\ell+\ell+1)}$ is a decreasing function in $k$ and $\ell$. In particular its maximum is $0.6$ for $k=3$ and 
$\ell=2$. Since $\beta > 0.6$ the above holds.
\end{proof}
We begin with setting $q_0 := 1 -{(\ell+1)(1-\beta ) \over k\ell}$ into $f(\beta, q_0)$ and 
obtain a function which depends only on $\beta$, namely
\begin{equation*}
\begin{split}
h(\beta) : =&\ln \bigg(\beta^{-(\ell+1)\beta} \cdot  \left(\left({(2^k-1)(\ell+1)\over {k\ell-(\ell+1)(1-\beta)} }\right)^{\ell}{k\ell-\xi^* \over k\ell-(1+\ell+\xi^*)(1-\beta)}\right)^{1-\beta}\left(1-{(\ell+1)(1-\beta) \over k\ell} \right)^{k\ell}\bigg).
\end{split}
\end{equation*}

\subsubsection*{Bounding $f(\beta,q)$ globally} 
To conclude the proof of Claim~\ref{cl:techn} it suffices to show that there exist $\eps_0$ and $C>0$ such that for any 
$\eps < \eps_0$ the following bounds hold
\begin{equation}\label{eq:3Stats} 
h(\beta), f(\beta, {1-(\ell+1)(1-\beta)/k\ell}),f(\beta, \beta)\le - C \eps,
\end{equation}
for all $0.6\leq \beta \leq 1-\eps$. These three inequalities will be shown in Claims~\ref{clm:h},~\ref{clm:fbb} and~\ref{clm:fb1-..},
respectively.

We will first bound $k\la -\xi^*$ which we will require to bound the above functions.
\begin{claim}\label{clm:techxi}
Let $k\ge 3,~\la \ge2$ and $\xi^*$ satisfies \eqref{eq:xiast}. Then $\xi^*>k\la -0.36$. Moreover, $k\la-\xi^* <0.19$ for $k=3,\la\ge 4$ and $k\ge 4,\la \ge2$. 
\end{claim}

\begin{proof}
Recall that $k\la=\frac{\xi^*Q(\xi^*,\ell)}{Q(\xi^*,\ell+1)}.$  
By definition we have 
\begin{equation} \label{eq:xidiff}
{k\ell \over \xi^*}={Q(\xi^*,\ell)\over Q(\xi^*,\ell+1)} =1+ {\Pr{\mathrm{Po}(\xi^*)=\ell} \over Q(\xi^*,\ell+1)} = 1+{1\over {\sum_{i\ge 1}{({\xi^{*}})^{i} \over (\la+1)\ldots (\la+i)}}}.
\end{equation}
Let \[\mathcal{S} := {\sum_{i\ge 1}{ {(\xi^{*})^i}\over (\la+1)\ldots (\la+i)}}\quad \textrm{and} \quad \mathcal{S}_i :={(\xi^{*})^{i}\over (\la+1)\ldots (\la+i)}. \]
Substituting $\xi^* = {k\la \over1+{1/ \mathcal{S}}} $ we obtain
\begin{equation} \label{eq:S_i}
 \mathcal{S}_i ={\left({1 \over1+{1 / \mathcal{S}}}\right)^{i} \over \left({1\over k}+{1\over k\la}\right)\ldots \left({1\over k}+{i\over k\la}\right)}.\end{equation}
By \eqref{eq:S_i} we have 
\begin{equation} \label{eq:S}
\mathcal{S} >\mathcal{S}_1={ k\la \cdot {\mathcal{S}  \over \mathcal{S} +1 }\over \la +1} \implies  \mathcal{S}   >{ k\la \over \la+1}-1\ge1.\end{equation}
So
$\xi^{*} = {k\la \over1+{1 / \mathcal{S}}} > {k\la \over 2}$ and thus $\xi^* \ge 3\la / 2$.
Therefore we obtain
\begin{align*}
 \mathcal{S} > {{k\la/ 2} \over \la+1} +  {({k\la / 2})^2 \over (\la+1)(\la+2)} +{({k\la / 2})^3 \over (\la+1)(\la+2)(\la+3)}.
\end{align*}
The right-hand side is clearly increasing in $k$ and $\la$. Therefore, substituting $k=3$ and $\la=2$ we obtain $\mathcal{S}>{2.2}$, implying that 
\begin{equation}\label{eq:valXi}
\xi^*> {(11/16)} k\la \geq {(33/ 16)} \la.
\end{equation}
In order to improve the bound upon $k\la -\xi^*$ we use the fact that $k\la - \xi^* =  {\xi^* / \mathcal{S}}$
and show that  ${\mathcal{S} \over \xi^*} >1$.
\begin{align*}
{\mathcal{S} \over \xi^*}& =  {\sum_{i\ge 1}{ {(\xi^{*})^{i-1}}\over (\la+1)\ldots (\la+i)}} ={1\over \la+1}\left( {\sum_{i\le \la}{ {(\xi^{*})^{i-1}}\over (\la+2)\ldots (\la+i)}}  +  {\sum_{i\ge \la+1}{ {(\xi^{*})^{i-1}}\over (\la+2)\ldots (\la+i)}}\right)\\
&\stackrel{\eqref{eq:valXi}}>{1\over \la+1} \left(\la +  {\sum_{i\ge \la+1}{ {(2\la)^{i-1}}\over (\la+2)\ldots (\la+i)}} \right).
\end{align*}
For $\la\ge 3$ observe that the term for $i=\la+1$ is
\[ 
 {{(2\la)^{i-1}}\over (\la+2)(\la+3)\ldots (2\la+1)}>  {2\la \cdot 2\la \over  (2\la-1) (2\la+1)}>1
\]
For $\la=2$ we have
\[  {\sum_{i\ge \la+1}{ {(2\la)^{i-1}}\over (\la+2)\ldots (\la+i)}} >\sum_{i=3}^5 {4^3 \over (2+i)(2+i-1)\ldots 5} >1.\]
By (\ref{eq:xidiff}), we have $k\la -\xi^* = {1\over {\sum_{i\ge 1}{({\xi^{*}})^{i-1} \over (\la+1)\ldots (\la+i)}}}$, and so
\[ {1\over k\la -\xi^*} > { {\sum_{i\ge \la +1}{({\xi^{*}})^{i-1} \over (\la+1)\ldots (\la+i)}}}> { {\sum_{i\ge \la +1}{({k\la-1})^{i-1} \over (\la+1)\ldots (\la+i)}}}.
\]
Let $S_i(k,\la) = { {{({k\la-1})^{i-1} \over (\la+1)\ldots (\la+i)}}}.$ Clearly $S_i(k,\la)$ is increasing with respect to $k$. Taking the derivative with respect to $\la$ we obtain that
\begin{align*}
\frac{\partial}{\partial \la}S_i(k,\la)
&= S_i(k,\la) \left(  {k(i-1)\over k\la-1} -{1\over \la+1}-{1\over \la+2}-\ldots -{1\over \la+i} \right) \\
&> S_i(k,\la) \left(  {i-1\over \la} -{1\over \la+1}-{1\over \la+2}-\ldots -{1\over \la+i} \right)\\
&= {S_i(k,\la) \over \la }\left(  {1\over \la+1} + {2\over \la+2} +\ldots + {i\over \la+i} -1 \right)\\
&>{S_i(k,\la) \over \la }\left( {i-1\over \la+i-1} +{i\over \la+i} -1\right)
\stackrel{i\ge {\la+1}}> {S_i(k,\la) \over \la }\left( {1\over 2} +{\la+1\over 2\la+1} -1\right) >0.
\end{align*}
Therefore, for all $i\ge \la+1$, $S_i(k,\la)$ increases with respect to $\la$. Numerical computations show that $(\sum_{i\ge\la+1}S_i(3,3))^{-1} <0.34$, $ ( \sum_{i\ge\la+1}S_i(3,4))^{-1} <0.15$ and $( \sum_{i\ge\la+1}S_i(4,2))^{-1} <0.19$. For the case $k=3,\la=2$ by direct computation we obtain $k\la-\xi^* <0.36$.
\end{proof}
\begin{claim} \label{clm:incXi}
For every $t \geq 1$, the function $x \to x Q(x,t-1)/Q(x,t)$ is increasing for $x > 0$.  
\end{claim}
\begin{proof} 
Set 
$$g_t(x) := {1\over (t-1)!}\cdot {1\over {1 \over t!} + {x \over (t+1)!} + {x^2 \over (t+2)!}+\cdots}.$$
Then
\begin{equation*}
\begin{split}
{xQ(x,t-1) \over Q(x,t)} &= {x(Q(x,t) + \Pr { \mathrm{Po} (x) = t-1)} \over Q(x,t)}
= x + g_t(x).
\end{split}
\end{equation*}
To see the claim it thus suffices to show that 
$$ -g_t'(x) < 1. $$
But 
$$- g_t'(x) = {1\over (t-1)!}~{{1 \over (t+1)!} + {2x \over (t+2)!} + {3x^{2} \over (t+3)!}+\cdots\over \left({1 \over t!} + {x \over (t+1)!}
+ {x^2 \over (t+2)!}+\cdots \right)^2}. $$

We, therefore, need to prove that
\begin{equation} \label{eq:ToProve} 
{1\over (t-1)!}\left({1 \over (t+1)!} + {2x \over (t+2)!} + {3x^{2} \over (t+3)!}+\cdots \right) < 
\left({1 \over t!} + {x \over (t+1)!} + {x^2 \over (t+2)!}+\cdots \right)^2. 
\end{equation}
We compare the coefficients on both sides one by one. Note that
$$ {1\over (t-1)! (t+1)!} < {1\over t!^2} \Leftrightarrow t < t+1. $$
Moreover,
$$ {2\over (t-1)! (t+2)!} < {2 \over t! (t+1)!} \Leftrightarrow t < t+2. $$
Next, the coefficient of $x^s$ for $s\geq 2$ on the right-hand side is
$$ 
\begin{cases}
2\sum_{i=0}^{\lfloor {s-1\over 2} \rfloor} {1\over (t+i)!(t+s-i)!} + {1 \over \left(t + \lceil {s-1 \over 2} \rceil \right)!^2},& \ \mbox{if} \ s \ \mbox{is even}, \\
2\sum_{i=0}^{\lfloor {s-1\over 2} \rfloor} {1\over (t+i)!(t+s-i)!},& \ \mbox{if} \ s \ \mbox{is odd} 
\end{cases}.
$$
Note that in any case we have (essentially) $s+1$ summands. So it suffices to show that each one of them is larger than the $1/(s+1)$th of
the coefficient of $x^s$ on the left-hand side, that is, ${1 \over (t-1)!(t+s+1)!}$. But this is the case, as for any $0\leq i \leq s$. 
$$ {1 \over (t-1)!(t+s+1)!} < {1\over (t+i)!(t+s-i)!} \Leftrightarrow (t+i)\cdots t < (t+s+1)\cdots (t+s-i+1) . $$
This now concludes the proof of the claim.
\end{proof}
We immediately obtain the following.
\begin{corollary}\label{clm:xiinc}
Let $k\ge 3,~\la \ge2$ and $\xi^*$ satisfies \eqref{eq:xiast}. Then $ \xi^* {Q(\xi^*,\ell)\over Q(\xi^*,\ell+1)}$ is increasing with respect to $\xi^*$. 
\end{corollary}
\begin{claim}\label{clm:ekl}
Let $e_{k,\la}$ be the value of derivative of ${xQ(x,\la)\over k\la \cdot Q(x,\la+1)}$ with respect to $x$ at $x=\xi^*$. Then 
$e_{k,\la} >{0.77\over \xi^* }$.
\end{claim}
\begin{proof}
We write 
\[
{xQ(x,\la) \over Q(x,\la+1)} = {x(Q(x,\la +1) + \Pr { \mathrm{Po} (x) = \la)} \over Q(x,\la+1)} = x + {{1\over \la!}\over {1 \over (\la+1)!} + {x \over (\la+2)!} + {x^2 \over (\la+3)!}+\cdots}.
\]
By definition
\begin{align*}
 e_{k,\la} \cdot k\la &= 1- {1\over \la!}~{{1 \over (\la+2)!} + {2\xi^* \over (\la+3)!} + {3{\xi^*}^{2} \over (\la+4)!}+\cdots\over \left({1 \over (\la+1)!} + {\xi^* \over (\la+2)!}
+ {{\xi^*}^2 \over (\la+3)!}+\cdots \right)^2} = 1- (k\la-\xi^*)\cdot {{1 \over (\la+2)!} + {2\xi^* \over (\la+3)!} + {3{\xi^*}^{2} \over (\la+4)!}+\cdots\over {1 \over (\la+1)!} + {\xi^* \over (\la+2)!}
+ {{\xi^*}^2 \over (\la+3)!}+\cdots }\\
&= 1- (k\la-\xi^*)\cdot \left(1- {{\la+1 \over (\la+2)!} + {(\la+1)\xi^* \over (\la+3)!} + {(\la+1){\xi^*}^{2} \over (\la+4)!}+\cdots\over {1 \over (\la+1)!} + {\xi^* \over (\la+2)!}
+ {{\xi^*}^2 \over (\la+3)!}+\cdots } \right)\\
&= 1- (k\la-\xi^*)\cdot \left(1- {\la+1 \over \xi^*}\cdot \left(1-{ {1\over (\la+1)!}\over {1 \over (\la+1)!} + {\xi^* \over (\la+2)!}
+ {{\xi^*}^2 \over (\la+3)!}+\cdots } \right)\right)\\
&= 1-(k\la-\xi^*) \left( 1- {\la+1\over \xi^*} + {k\la-\xi^* \over \xi^*} \right)=1-(k\la-\xi^*) \left(- {\la+1\over \xi^*} + {k\la \over \xi^*} \right).
\end{align*}
Thus,
\begin{align*}
e_{k,\la} &= {1\over k\la } - {k\la-\xi^*\over k\la} \left(- {\la+1\over \xi^*} + {k\la \over \xi^*} \right)
= {1\over k\la } + {\la+1 \over \xi^*} -{\la+1\over k\la}- {k\la-\xi^* \over \xi^*} \\
&=  {1\over \xi^*} - {k\la-\xi^* \over \xi^*}  +{k\la-\xi^* \over \xi^*k} .
\end{align*}
One can check that for $k=3,\la=2$, $e_{k,\la} >{ 0.77\over \xi^*}$ and for $k=3,\la=3$, $e_{k,\la} >{ 0.89\over \xi^*}$.
For other values we use
\[ e_{k,\la}\cdot \xi^*  > 1-(k\la-\xi^*) .\]
which by second part of Claim~\ref{clm:techxi} is at least $0.81$.
\end{proof}
\begin{claim} \label{clm:xi}
For any $k\geq 3$ and $\ell\geq 2$ we have $\xi^* <k\la$ and 
 \[\xi^* >k\ell-{e^{-k\ell} (k\ell)\cdot (k\ell-0.36)^{\ell} \over \ell!
 }\left(1- \exp{\left(-(k\ell-\ell +0.64)^2\over 2k\ell -0.72\right)} \right)^{-1}.\]
\end{claim}

\begin{proof}
We have $k\cdot \ell= \xi^* \cdot {Q(\xi^*,\ell)\over Q(\xi^*,\ell+1)}.$ As ${Q(\xi^*,\ell)\over Q(\xi^*,\ell+1)} > 1$ for all $\xi^*$ and $\ell$, we deduce that $\xi^* < k\ell$. 
By Claim~\ref{clm:techxi}  we know that for all $k\ge 3$ and $\la\ge 2$, $\xi^* >  k\la -0.36.$ In order to improve upon the above bound, note first that
\begin{equation} \label{eqn:xi1}
\begin{split}
\xi^*=  k\ell \cdot {Q(\xi^*,\ell +1)\over Q(\xi^*,\ell)}= k\ell-k\la{\Pr{\mathrm{Po}(\xi^*)=\ell}\over Q(\xi^*,\ell) }
\ge  ~~k\ell-k\la{\Pr{\mathrm{Po}(k\ell-0.36)=\ell}\over Q(k\ell-0.36,\ell) } .
\end{split}
\end{equation}
Let $X$ be a Poisson random variable with parameter $\mu= k\ell-0.36$ . Thus, $Q(k\ell-0.36,\ell)=1-\Pr{X \le \ell-1}.$ We define $\delta = 1-(\ell-1)/\mu $.
Now, for any $t<0$ we have 
\begin{align*}
\Pr{X\le \ell-1} = &\Pr{X\le (1-\delta)\mu}=\Pr{e^{tX}\ge e^{t(1-\delta)\mu}}\\
& \le {\E{e^{tX}}\over e^{t(1-\delta)\mu}} = {\exp(-\mu +\mu \cdot e^{t})\over \exp(t(1-\delta)\mu)}.
\end{align*}
Setting $t=\ln(\ell-1)-\ln(\mu)$ we have
\begin{equation}\label{eqn:xi2II}
\Pr{X \le \ell-1} < \left(e^{-\delta}\over (1-\delta)^{(1-\delta)}\right)^{\mu} < \exp{\left(-(\mu-\ell +1)^2\over 2\mu\right)}.
\end{equation}
The combination of~\eqref{eqn:xi1} and~\eqref{eqn:xi2II} lead us to the stated lower bound.
\end{proof}
In what follows we use the following definition 
$$t(k,\ell): = \left(1-{0.36\over k\ell}\right)^{\ell} \left(1-\exp{\left(-(k\ell-\ell +0.64)^2\over 2k\ell-0.72\right)}\right)^{-1}.$$
We are now ready to deduce the inequalities in~(\ref{eq:3Stats}), starting with a bound on $h(\beta)$.
%
\begin{claim} \label{clm:h}
For any $k\geq 3$ and $\ell\geq 2$ there is a $C_1 > 0$ such that for any $0< \eps < 1$ and any $0.6 \le \beta \leq 1-\eps$ we have 
$h(\beta)\le -C_1\varepsilon $.
\end{claim}
\begin{proof} 
By Claim~\ref{clm:xi}, we have $k\ell-t(k,\ell)\cdot {e^{-k\ell} (k\ell)^{\ell+1} \over \ell!} < \xi^* < k\ell$. Using these bounds for $\xi^*$ we obtain
\begin{equation}
\label{eq:conf}
\begin{split}
 e^{h(\beta)} < &~\beta^{-(\ell+1)\beta}\left({(2^k-1)(\ell+1)\over k\ell-(\ell+1)(1-\beta)}\right)^{\ell(1-\beta)} \\
&\times \left( t(k,\ell)\cdot {e^{-k\ell} (k\ell)^{\ell+1} \over \ell!}\over k\ell-(\ell+k\ell+1)(1-\beta)\right)^{1-\beta} 
\left(1-{(\ell+1)(1-\beta) \over k\ell} \right)^{k\ell} \\
=&\left({2^k-1 \over e^{k}\cdot \beta^{\beta\over (1-\beta)}}\right)^{\ell(1-\beta)} 
 \left(1-{(\ell+1)(1-\beta)\over k\ell}  \right)^{-\ell(1-\beta)} 
\cdot\left(1-{(\ell+k\ell+1)(1-\beta)\over k\ell}  \right)^{-(1-\beta)}\\
 &\times \left({(\ell+1)^{\ell}\cdot t(k,\ell)}\over \beta^{\beta\over (1-\beta)} \la! \right)^{1-\beta}
\left(1-{(\ell+1)(1-\beta) \over k\ell} \right)^{k\ell}.
\end{split}
\end{equation}
Using the inequality $(1-x)^{-1}\leq \exp\left({x+{x^2\over 1.4}}\right)$ for $x\leq 0.4$ we can deduce  
\begin{equation}\label{eq:betaineq}\beta^{-\beta \over 1-\beta}=\left(1-(1-\beta)\right)^{-\beta \over 1-\beta} \leq  e^{\beta +{(1-\beta)\beta \over 1.4}}.\end{equation} 
Also,
\begin{align*}
&\left(1-{(\ell+1)(1-\beta)\over k\ell}  \right)^{-1}
\leq\exp\left\lbrace{{(\ell+1)(1-\beta)\over k\ell}+ {(\ell+1)^2(1-\beta)^2 \over 1.4 (k\ell)^2}}\right\rbrace,\\
&\left(1-{(1+ \ell+ k\ell)(1-\beta)\over k\ell}  \right)^{-1/\ell}
\leq \exp\left\lbrace{{(1-\beta) (1+\ell +k\ell)\over k\ell^2}+ {(1-\beta)^2 (1+\ell+k\ell)^2\over k^2\ell^3}}\right\rbrace ,\\
&\left(1- {(\ell+1)(1-\beta)\over k\ell}\right)^{k\ell} <\exp\left(-(\ell+1)(1-\beta)- {(\la +1)^2(1-\beta)^2 \over 2k\la}\right).
 \end{align*}
 By Stirling's formula and (\ref{eq:betaineq}) we have
\[
{{(\la+1)^{\la} }\over \la!\cdot \beta^{\beta \over 1-\beta}} < {(1+1/\ell)^{\ell} \exp(\la)\over \sqrt{2\pi\ell}}  \exp\left(\beta+{\beta(1-\beta )\over 1.4}\right).
\] 
Now combining the last two terms in~\eqref{eq:conf} we obtain
\begin{align*}
&\left({(\ell+1)^{\ell}\cdot t(k,\ell)}\over \beta^{\beta\over (1-\beta)} \la! \right)^{1-\beta}
\left(1-{(\ell+1)(1-\beta) \over k\ell} \right)^{k\ell} \\
&~~~~~~~~~~~~< \left( {(1+1/\ell)^{\ell}\cdot t(k,\la)\over \sqrt{2\pi\ell}}\right)^{1-\beta}  \exp\left(\beta(1-\beta)+{\beta(1-\beta )^2\over 1.4}- (1-\beta) - {(\la+1)^2(1-\beta)^2\over 2k\la}\right)\\
 &~~~~~~~~~~~= \left( {(1+1/\ell)^{\ell}\cdot t(k,\la)\over \sqrt{2\pi\ell}}\right)^{1-\beta}  \exp\left(\beta(1-\beta)+{\beta(1-\beta )^2\over 1.4}-(1-\beta)-\left( 1+{1\over \la}\right){(\la+1)(1-\beta)^2\over 2k}\right).
\end{align*}
Also recall that
\[
t(k,\ell) = \left(1-{0.36\over k\ell}\right)^{\ell} \left(1-\exp{\left(-(k\ell-\ell +0.64)^2\over 2k\ell-0.72\right)}\right)^{-1}.\]
Substituting these bounds  in~\eqref{eq:conf} we obtain
\begin{equation}\label{eq:e(h(beta)}
 \begin{split}
e^{h(\beta)} < &\left( \left({2^k-1 \over \exp{\left(k- \Delta_{k,\ell,\beta} \right)}  } \right)^{\la} \cdot {{(1+1/\ell)^{\ell}}  \exp\left(\beta +{\beta(1-\beta )\over 1.4}-1\right)\over\sqrt{ 2\pi\la}\cdot \left(1-\exp{\left(-(k\ell-\ell +0.64)^2\over 2k\ell-0.72\right)}\right)}\right)^{1-\beta},
 \end{split}
\end{equation}
where
\begin{equation*}
 \begin{split}
 \Delta_{k,\ell,\beta}:=&\beta +{(1-\beta)\beta \over 1.4}+{(\ell+1)(1-\beta)\over k\ell}+ {(\ell+1)^2(1-\beta)^2 \over 1.4 (k\ell)^2}+ {(1-\beta) (1+k\ell+\ell)\over k\ell^2}\\
 &+{(1-\beta)^2 (1+k\ell+\ell)^2\over k^2\ell^3} -\left(1+{1\over \la}\right){(\la+1)(1-\beta) \over 2k\la} \\
 =&\beta +{(1-\beta)\beta \over 1.4}+{(\ell+1)(1-\beta)\over 2k\ell}+ {(\ell+1)^2(1-\beta)^2 \over 1.4(k\ell)^2}+ {(1-\beta) (1+k\ell+\ell)\over k\ell^2} \\
 &+{(1-\beta)^2 (1+k\ell+\ell)^2\over k^2\ell^3} - {1\over \la}~{(\la+1)(1-\beta) \over 2k\la} \\
 =&\beta +{(1-\beta)\beta \over 1.4}+{(1+1/\ell)(1-\beta) \over2 k}+ {(1+1/\ell)^2(1-\beta)^2 \over 1.4~k^2} + {(1-\beta) (1/2k\ell+1+1/2k)\over \ell}\\ &+{(1-\beta)^2 (1/k\ell+1+1/k)^2\over \ell}.
 \end{split}
 \end{equation*}
 We note that $\Delta_{k,\la,\beta}$ is decreasing in $k$ and $\la$. The partial derivative of $\Delta_{k,\ell,\beta}$ with respect to $\beta$  is given by
\begin{align*}
\Delta_{k,\ell,\beta}':={\partial{\Delta_{k,\ell,\beta}} \over \partial{\beta}}=& {12\over 7} - {10\over 7}\beta -{1+1/\ell\over2 k}- {(1+1/\ell)^2(1-\beta) \over (0.7) k^2}- {1/2k\ell+1+1/2k\over \ell}\\
 &-{2(1-\beta) (1/k\ell+1+1/k)^2\over \ell}.\end{align*}
 Observe that ${\partial{\Delta_{k,\ell,\beta}} \over \partial{\beta}}$ is increasing with $k$ and $\la$.
%
Let 
\begin{align*}
p(k,\ell,\beta):= &\left({2^k-1 \over \exp{\left(k- \Delta_{k,\ell,\beta} \right)}  }  \right) \quad and \quad g(k,\ell):={\exp(1) \over\sqrt{ 2\pi\la}\cdot \left(1-\exp{\left(-(k\ell-\ell +0.64)^2\over 2k\ell-0.72\right)}\right)}.
\end{align*}
One can check that 
\[e^{h(\beta)} <((p(k,\la,\beta))^{\la} g(k,\la))^{1-\beta}.\]
We start with the case $k\geq 4$. Firstly note that $\Delta'_{4,2,\beta}=-519/448+(297/448)\beta$ which is negative for all $\beta<1$. 
Also, as ${(2^k-1)\cdot\exp(-k)}$ is decreasing in $k$ and $\Delta_{k,\la,\beta}$ is decreasing in $k$ and $\la$ we infer that for $k\ge 4,
\la\ge 2$, thus the maximum value of $p(k,\la,\beta)$ is $p(4,2,0.6).$ Numerical computations show that $p(4,2,0.6) <0.97.$ Now, clearly
$g(k,\la)$ is decreasing in $k$ and $\la$. Moreover, one can check that $g(3,2)<0.91$, which completes the proof for $k\ge 4,\la\ge 2$.

For the case $k=3$, firstly note that $\Delta'_{3,5,\beta} =229/875-(52/125)\beta$, which implies that  $\Delta_{3,5,\beta}$ is maximized at
$\beta=\beta_{max}=229/364$. Therefore, for $\la \ge 5$, $p(3,\la,\beta)$ is maximized at $p(3,5,\beta_{max})$. 
Numerical computations show that $p(3,5,\beta_{max})<0.98.$

For the cases $\la \le 4$ ,  firstly note that $\Delta'_{3,4,\beta} =-1/21-17\beta /96 \stackrel{\beta>0}<0$. Now let
\[m(k,\la,\beta) := p(k,\la,\beta)^{\la} g(k,\la).\]

Recall that $\Delta'_{k,\la,\beta} $ is increasing in $k$ and $\la$. Also, $\Delta_{3,4,\beta} $ is decreasing in $\beta$. We can therefore conclude that for all $\beta\ge 0.6$ and $\la\le4$, $m(3,\la,\beta) \le m(3,\la,0.6) $. One can check that
$m(3,3,0.6) < 0.93$ and $m(3,4,0.6) <0.62$. The case $\la=2$ is more tedious. We substitute $k=3,\la=2$ in~\eqref{eq:e(h(beta)}. 
\begin{equation}
 \begin{split}
e^{h(\beta) \over 1-\beta} < &\left({7\over \exp{\left(3- \Delta_{3,2,\beta} \right)}  } \right)^{2} \cdot {{(1+1/2)^{2}}  \exp\left(\beta+{\beta(1-\beta )\over 1.4}-1\right)\over\sqrt{ 4\pi}\cdot \left(1-\exp{\left(-(4.64)^2\over11.28\right)}\right)}\\
<&\left({7 \over \exp{\left(3- \Delta_{3,2,\beta} -{{\beta\over2} -{\beta(1-\beta)\over2.8}}\right)}  } \right)^{2} \cdot { 2.25\cdot \exp\left(-1\right)\over\sqrt{ 4\pi}\cdot \left(1-\exp{\left(-(4.64)^2\over11.28\right)}\right)}
\end{split}
\end{equation}
Now we check that the partial derivative of $ \Delta_{3,2,\beta} +{{\beta\over2} +{\beta(1-\beta)\over2.8}}$ with respect to $\beta$ is less than $-0.91+0.47\beta$, which implies that the right-hand side is decreasing with respect to $\beta$ for $\beta\le 1$. We complete the proof by calculating the above expression for $\beta=0.6$ which gives $e^{h(\beta)}<(0.91)^{1-\beta}$.
\end{proof}

\begin{claim} \label{clm:fbb}
For any $k\ge3$ and $\ell\ge2$ there exist $\eps_0 > 0$ and $C_2 >0$ such that the following holds. For any $\eps < \eps_0$, if 
$0.6<\beta\le 1-\varepsilon$ we have $$f(\beta,\beta)<-C_2 \eps.$$
\end{claim}
\begin{proof}
By Lemma~\ref{lem:I}, it follows that substituting $q=\beta$ in ${k\la(1-q) \over 1-\beta}$ we have
\[
	I_{\xi^*}\left(\frac{k\ell(1-\beta)}{1-\beta}\right)= 0.
\]
So,
\begin{equation*}
\begin{split}
f(\beta, \beta) = -(k\ell-\ell-1)H(\beta) + \ell(1-\beta) \ln \left(2^k  -1 \right).
\end{split}
\end{equation*} 
Note that for any~$k\ge 3$ and $\ell\ge 2$ this function is convex with respect to~$\beta$, as~$-H(\beta)$ is convex and the linear term that
is added preserves its convexity.  Note that $-H(1-\eps) < - \eps \ln (1/\eps)$, whereby it follows that there exists a constant 
$C_2 = C_2(k,\ell) >0$ such that for any $0< \eps <1/e$ we have 
$$ f(1-\eps, 1-\eps) < - C_2 \eps \ln (1/\eps)  < -C_2 \eps. $$

Since~$H(0.6)>0.6$, we have 
$$ f(0.6,0.6)< -0.6(k\ell-\ell-1) + 0.4\ell \ln \left(2^k  -1 \right).$$
The derivative of this function with respect to~$k$ is 
$-0.6\ell + \ell  \cdot 0.4 {2^k \ln 2  \over 2^k -1}$. A simple calculation shows that the second summand is less than~$0.32\ell$ for 
all~$k\geq 3$. The derivative with respect to $\ell$ is $-0.6k+0.6 +0.4\ln(2^k-1)$ which is again a decreasing function in $k$ and less than
$-0.42$ at $k=3$. So, we may set~$k=3$ and~$\ell=2$, thus obtaining $f(0.6,0.6) < -1.8 + 0.8\ln 7 < -0.24$. The above analysis along
with the convexity of~$f(\beta, \beta)$  imply the claimed statement.
\end{proof}
\begin{claim} \label{clm:fb1-..}
For all $k\geq 3$ and $\ell\geq 2$ there is a $C_3>0$ such that for all $\eps$ and for all $\beta \leq 1-\eps$ 
$$f(\beta, 1- (\ell+1)(1-\beta)/k\ell) \le - C_3 \eps.$$ 
\end{claim}
\begin{proof}
Substituting $1- (\ell+1)(1-\beta)/k\ell$ for $q$ into the formula of $f$ we obtain:
\begin{equation*}
\begin{split}
f\left(\beta, 1 - {(\ell+1)(1-\beta) \over k\ell} \right)
= &(\ell+1)H(\beta)+ \ell(1-\beta ) \ln (2^k  -1)\\
&- k \la H\left( {k\ell-(\ell+1)(1-\beta) \over k\ell}\right)-(1-\beta) I(\ell+1). 
\end{split}
\end{equation*}
Note that for $\beta = 1$ the expression is equal to 0. 
To deduce the bound we are aiming for, 
we will show that in fact $f\left(\beta, 1-{(\ell+1)(1- \beta)/ k\ell} \right)$ is an increasing function with 
respect to $\beta$. That is, we will show that its first derivative with respect to $\beta$ is positive for any $\beta \leq 1$. 
Finally, Taylor's Theorem around $\beta = 1$  implies the claim. 

We get
\begin{equation*}
\begin{split} 
{\partial f\left(\beta, 1- {(\ell+1)(1- \beta) \over k\ell} \right) \over \partial \beta}=&(\ell+1) \ln \left({1-\beta \over \beta} \right) - \ell\ln (2^k -1) \\
&- (\ell+1) \ln \left( {(\ell+1)(1-\beta) \over k\ell-(\ell+1)(1-\beta)} \right) + I(\ell+1). 
\end{split}
\end{equation*}
Substituting for $I(\ell+1)$ the value given in Lemma~\ref{lem:I} and since $e^{\xi}Q(\xi,\ell+1) = \xi^{\ell+1}/\ell!(k\ell-\xi)$ we obtain for $\beta<1$
\begin{equation*}
\begin{split} 
{\partial f\left(\beta,1 - {(\ell+1)(1- \beta ) \over k\ell} \right) \over \partial \beta}
= \ln \left(\left(k\ell-(\ell+1)(1-\beta) \over (\ell+1)\beta \right)^{\ell+1}(2^k-1)^{-\ell}\cdot{\ell+1\over k\ell-\xi}\right).
\end{split}
\end{equation*}
We will show that the fraction inside the logarithm is greater than 1. 
Note first that \[{k\ell-(\ell+1)(1-\beta)\over (\ell+1)\beta} = {1\over \beta} \left (k\la -(\la+1) \over \la +1\right) +1 = {1\over \beta} \left ((k-1)\la -1 \over \la +1\right) +1\]
is decreasing with respect to $\beta$ -- so we obtain a lower bound by setting~$\beta =1$. Substituting $\beta=1$ we obtain
\[
{\partial f\left(\beta,1 - {(\ell+1)(1- \beta ) \over k\ell} \right) \over \partial \beta}
> \ln \left(\left(k\la \over \ell+1\right)^{\ell+1}(2^k-1)^{-\ell}\cdot{\ell+1\over k\ell-\xi}\right).
\]
By Claim~\ref{clm:xi}, for all $k\geq3$ and $\ell\geq 2$ we have $k\ell-\xi \le {e^{-k\ell} (k\ell)^{\ell+1}\over \ell! (1- e^{-(k\ell-\ell+0.64)^2 / 2k\ell-0.72})}$ which yields
\begin{equation}
\label{eq:tmpFinal1}
\begin{split}
\left(k\ell\over \ell+1\right)^{\ell+1}(2^k-1)^{-\ell}\cdot{(\ell+1)\over k\ell-\xi}
\geq& {e^{k\ell} \ell! (1- e^{-(k\ell-\ell+0.64)^2 / 2k\ell-0.72})\over (2^k-1)^{\ell} (\ell+1)^{\ell} }\\
=&  {e^{k\ell} \ell! (1- e^{-(k\ell-\ell+0.64)^2 / 2k\ell-0.72})\over \la^\la(2^k-1)^{\ell} (1+1/\ell)^{\ell} }\\
\stackrel{1+x\leq e^x}>& { \ell!  \over e\cdot\la^\la}\cdot{e^{k\ell}(1- e^{-(k\ell-\ell+0.64)^2 / 2k\ell-0.72})\over (2^k-1)^{\ell}}
\end{split}
\end{equation}
Using the bounds $\ell!\geq \sqrt{2\pi \ell} (\ell/e)^{\ell}$ and $1+x\leq e^x$ we can further bound the right-hand side 
of~\eqref{eq:tmpFinal1} as follows:
\begin{equation}
\label{eq:tmpFinal}
\begin{split}
 { \ell!  \over e\cdot\la^\la}\cdot{e^{k\ell}(1- e^{-(k\ell-\ell+0.64)^2 / 2k\ell-0.72})\over (2^k-1)^{\ell}}
~\geq{\sqrt{2\pi \ell} \over e^{\ell+1} }\cdot {e^{k\ell} (1- e^{-(k\ell-\ell+0.64)^2 / 2k\ell-0.72})\over (2^k-1)^{\ell} }.
\end{split}
\end{equation}
It is easy to verify that $\sqrt{2\pi \ell}(1- e^{-(k\ell-\ell+0.64)^2 / 2k\ell-0.72})$ is increasing in $k$ and $\ell$. Also the first derivative of 
the function $e^{k}/ (2^k-1)$ with respect to $k$ is $e^k (2^k(1-\ln(2))-1)/(2^k-1)^2$ which is positive for any $k\geq 3$. Moreover the first
derivative of the function $e^{k\ell-\ell-1}/(2^k-1)^{\ell}$ with respect to $\ell$ is $e^{k\ell-\ell-1}(2^k-1)^{-\ell}(k-\ln(2^k-1)-1)$ which is
positive for any $k\geq 3$ and $\ell\geq 2$. So we infer that the right-hand side of the above inequality is increasing in both $k$ and $\ell$.
Numerical calculations show that the right hand side of the above inequality is greater than $1.2$ for $k=3 , \ell=2$. The above arguments
establish the fact that the derivative of $f\left(\beta, 1 -{(\ell+1)(1- \beta) / k\ell} \right)$ with respect to $\beta$ is positive, for all 
$k\geq 3$ and $\ell\geq2$. 
\end{proof}

\small
\bibliographystyle{plain}
\bibliography{l-Orient}

%

\end{document}